\documentclass[11pt]{article}
\usepackage{amsfonts}
\usepackage{amsmath}
\usepackage{epsfig}
\usepackage{mathrsfs}
\usepackage{graphicx}
\usepackage{amssymb}
\usepackage[scriptsize]{subfigure}

\newtheorem{theorem}{Theorem}

\newtheorem{definition}{Definition}
\newtheorem{lemma}{Lemma}
\newtheorem{corollary}{Corollary}

\newenvironment{proof}[1][Proof]{\noindent \textbf{#1.} }{\qedsymbol}
\newcommand{\qedsymbol}{\hspace{\fill}\rule{1.5ex}{1.5ex}}
\input{tcilatex}
\def\baselinestretch{1.25}
\def\b0{\mbox{\boldmath $0$}}

\begin{document}

\title{{\Huge {Distributed Decision Through Self-Synchronizing%
}}\\
\vspace{-0.1cm} {\Huge Sensor Networks in the Presence of Propagation}\\
\vspace{-0.1cm} {\Huge Delays and Asymmetric Channels}}
\author{Gesualdo Scutari, Sergio Barbarossa and Loreto Pescosolido \\
{\small Dpt. INFOCOM, Univ. of Rome \textquotedblleft La Sapienza
\textquotedblright, Via Eudossiana 18, 00184 Rome, Italy}\\
{\small E-mail: \texttt{$\{$scutari, sergio, loreto$\}$@infocom.uniroma1.it}}%
.\thanks{%
This work has been partially funded by the WINSOC project, a Specific
Targeted Research Project (Contract Number 0033914) co-funded by the INFSO
DG of the European Commission within the RTD activities of the Thematic
Priority Information Society Technologies, and by ARL/ERO Contract
N62558-05-P-0458. \newline
Part of this work was presented in the \textit{IEEE Workshop on Signal
Processing Advances in Wireless Communications, (SPAWC-2006)}, July 2006,
for undirected graphs.}}
\date{{\small Manuscript submitted on January 10, 2007. Finally accepted on September 15, 2007.}}
\maketitle
\vspace{-0.9cm}
\begin{abstract}
In this paper we propose and analyze a distributed algorithm for
achieving globally optimal decisions, either estimation or
detection, through a self-synchronization mechanism among linearly
coupled integrators initialized with local measurements. We model
the interaction among the nodes as a directed graph with weights
(possibly) dependent on the radio channels and we pose special
attention to the effect of the propagation delay occurring in the
exchange of data among sensors, as a function of the network
geometry. We derive necessary and sufficient conditions for the
proposed system to reach a consensus on globally optimal decision
statistics. One of the major results proved in this work is that a
consensus is reached with exponential convergence speed for any
bounded delay condition if and only if the directed graph is
quasi-strongly connected. We provide a closed form expression for
the global consensus, showing that the effect of delays is, in
general, the introduction of a bias in the final decision. Finally,
we exploit our closed form expression to devise a double-step
consensus mechanism able to provide an unbiased estimate with
minimum extra complexity, without the need to know or estimate the
channel parameters.
\end{abstract}

\vspace{-0.5cm}

\section{Introduction and Motivations}

Endowing a sensor network with self-organizing capabilities is undoubtedly a
useful goal to increase the resilience of the network against node failures
(or simply switches to sleep mode) and avoid potentially dangerous
congestion conditions around the sink nodes. Decentralizing decisions
decreases also the vulnerability of the network against damages to the sink
or control nodes. Distributed computation over a network
and its application to statistical consensus theory has a long
history (see, e.g., \cite{Eisenberg-Gale59, DeGroot74}), including
the pioneering work of Tsitsiklis, Bertsekas and Athans on
asynchronous agreement problem for
discrete-time distributed decision-making systems \cite%
{Tsitsiklis-Bertsekas-Athans} and parallel computing \cite%
{Tsitsiklis-thesis, Tsitsiklis-Bertsekas-book}. A sensor network may
be seen indeed as a sort of distributed computer that has to
evaluate a function of the measurements gathered by each sensor,
possibly without the need of a fusion
center. 
This problem may be addressed by taking into account the vast literature on
distributed consensus/agreement algorithms. These techniques have received
great attention in the recent years within the literature on cooperative
control and multiagent systems \cite{Borkar-Varaiya}$-$\cite%
{Blondel-Tsitsiklis}.
In particular, the conditions for achieving a consensus over a
common specified value, like a linear combination of the
observations, was solved for networked \emph{continuous-time}
dynamic systems by Olfati-Saber and Murray, under a variety of
network topologies, also allowing for topology
variations during the time necessary to achieve consensus \cite%
{Olfati-Saber-2003, Olfati-Saber}. The \emph{discrete-time} case was
addressed by Tsitsiklis in \cite{Tsitsiklis-thesis} (see also \cite%
{Tsitsiklis-Bertsekas-book}). Many recent works focused on the distributed
computation of more general functions than the average of the initial
measurements. These include average-max-min consensus \cite%
{Corts-Automatica-06}, geometric mean consensus \cite{Olfati-Franco}, and
power mean consensus \cite{Bauso-Giarre}. A study on the class of smooth
functions that can be computed by distributed consensus algorithms was
recently addressed in \cite{Corts-Automatica-sub}. Application of consensus
algorithms to data fusion problem and distributed filtering was proposed in
\cite{Olfati-Shamma, Olfati-Kalman}. The study of the so called \textit{%
alignment} problem, where all the agents eventually reach an agreement, but
without specifying how the final value is related to the initial
measurements, was carried out in \cite{Jadbabaie03}-\cite{Blondel-Tsitsiklis}%
. Two recent excellent tutorials on distributed consensus and agreement
techniques are \cite{Olfati-Saber-Murray-ProcIEEE07} and \cite%
{Ren-Beard-Control-Magazine}.

Consensus may be also seen as a form of self-synchronization among coupled
dynamical systems. In \cite{Barbarossa-iwwan05, Barbarossa-Scutari-Journal},
it was shown how to use the self-synchronization capabilities of a set of
nonlinearly coupled dynamical systems to reach the \textit{globally optimal}
Best Linear Unbiased Estimator (BLUE), assuming symmetric communication
links. As well known, the BLUE estimator coincides with the Maximum
Likelihood (ML) estimate for the linear observation model with additive
Gaussian noise. In particular, it was shown in \cite%
{Barbarossa-Scutari-Journal} that reaching a consensus on the state
derivative, rather than on the state itself (as in \cite{Olfati-Saber-2003}$%
- $\cite{Ren-Beard-Control-Magazine}), allows for better resilience
against coupling noise. In \cite{Schizas-Ribeiro-Giannakis}, the
authors provided a discrete-time decentralized consensus algorithm
to compute the BLUE, based on linear coupling among the nodes, as a
result of a distributed optimization incorporating the coupling
noise.

The consensus protocols proposed in \cite{Olfati-Saber-2003}, \cite%
{ren-et-al}$-$\cite{Jin-Murray} and \cite{Ren-Beard-Control-Magazine}$-$\cite%
{Schizas-Ribeiro-Giannakis} assume that the interactions among the nodes
occur instantaneously, i.e., without any propagation delay. However, this
assumption is not valid for large scale networks, where the distances among
the nodes are large enough to introduce a nonnegligible communication delay.
There are only a few recent works that study the consensus or the agreement
problem, in the presence of propagation delays, namely \cite{Olfati-Saber},
\cite{Strogatz}$-$\cite{Scutari-Barbarossa-Delay-SPAWC}, \cite{Chellaboina-Haddad-CDC06} focusing on \emph{%
continuous-time} systems, and \cite{Tsitsiklis-Bertsekas-Athans,
Tsitsiklis-thesis}, \cite[Ch. 7.3]{Tsitsiklis-Bertsekas-book}, \cite%
{Blondel-Tsitsiklis}, dealing with \emph{discrete-time} protocols.
More specifically, in \cite{Olfati-Saber, Strogatz} the authors
provide sufficient conditions for the convergence of a linear
average consensus protocol, in the case of time-invariant
\emph{homogeneous} delays (i.e., equal delay for all the links) and
assuming \emph{symmetric} communication links. The most appealing
feature of the dynamical system in \cite{Olfati-Saber} is the
convergence of all the state variables to a \emph{known}, \emph{%
delay-independent}, function (equal to the average) of the
observations. However, this desired property is paid in terms of
convergence, since, in the presence of homogeneous delays, the
system of \cite{Olfati-Saber} is able to reach a consensus if and
only if the (common) delay is smaller than a given,
topology-dependent, value. Moreover, the assumption of homogeneous
delays and symmetric links is not appropriate for modeling the
propagation delay in a typical wireless network, where the delays
are proportional to the distance among the nodes and communication
channels are typically \emph{asymmetric}.

The protocol of \cite{Olfati-Saber} was generalized in \cite%
{Papachristodoulou-CDC06} to time-invariant \emph{inhomogeneous}
delays (but symmetric channels) and in \cite{Lee-Spong-06} to
asymmetric channels. The dynamical systems studied in
\cite{Papachristodoulou-CDC06, Lee-Spong-06} are guaranteed to reach
an agreement, for any given set of finite propagation delays,
provided that the network is \textit{strongly connected}. Similar
results, under weaker (sufficient) conditions on the (possibly
time-varying) network topology, were obtained in
\cite{Tsitsiklis-Bertsekas-Athans, Tsitsiklis-thesis}, \cite[Ch.7.3]
{Tsitsiklis-Bertsekas-book} and \cite{Blondel-Tsitsiklis} for the
convergence of discrete-time asynchronous agreement algorithms.
However, since the final agreement value is not a known function of
the local measurements, agreement algorithms proposed in the cited
papers are mostly appropriate for the alignment of mobile agents,
but they cannot be immediately used to distributively compute
prescribed functions of the sensors' measurements, like decision
tests or global parameter estimates. In
\cite{Chellaboina-Haddad-CDC06}, the authors studied the convergence
properties of the agreement algorithms proposed in
\cite{Papachristodoulou-CDC06}$-$\cite{Wang-Xiao-06} and provided a
closed form expression of the achievable consensus in the case of
strongly connected balanced digraphs, under the assumption that the
initial conditions are nonnegative and the state trajectories remain
in the nonnegative orthant of the state space. However, the final
agreement value depends on delays, network topology and initial
conditions of each node, so that the bias is practically
unavoidable. In summary, in the presence of propagation delays,
classical distributed protocols
reaching consensus/agreement  on the state \cite%
{Borkar-Varaiya}$-$\cite{Ren-Beard-Control-Magazine}, \cite%
{Schizas-Ribeiro-Giannakis}$-$\cite{Wang-Xiao-06},
\cite{Chellaboina-Haddad-CDC06} cannot be used to achieve prescribed
functions of the sensors' measurements that are not biased by the
channel parameters.

Ideally, we would like to have a totally decentralized system able
to reach a global consensus on a final value which is a
\emph{known}, \emph{delay-independent}, function of the local
measurements (as in \cite{Olfati-Saber}), for
\emph{any} given set of \emph{inhomogeneous} propagation delays (as in \cite%
{Papachristodoulou-CDC06}) and \emph{asymmetric} channels (as in \cite%
{Lee-Spong-06}). In this paper, we propose a distributed dynamical
system having all the above desired features. More specifically, we
consider a set of linearly coupled first-order dynamical systems and
we fully characterize its convergence properties, for a network with
arbitrary topology (not necessarily strongly
connected, as opposed to \cite{Papachristodoulou-CDC06}$-$\cite{Wang-Xiao-06}%
) and (possibly) asymmetric communication channels. In particular,
we consider an interaction model among the sensors that is directly
related to the physical channel parameters. The network is modeled
as a weighted directed graph, whose weights are directly related to
the (flat-fading) channel coefficients between the nodes and to the
transmit power.  Furthermore, the geometry-dependent propagation
delays between each pair of nodes, as well as possible time offsets
among the nodes, are properly taken into account. The most appealing
feature of the proposed system is that a consensus on a globally
optimal decision statistic is achieved, for \textit{any} (bounded)
set of inhomogeneous delays and for \textit{any} set of (asymmetric)
communication channels, with the only requirement that the network
be
quasi-strongly connected and the channel coefficients be nonnegative.%
\footnote{%
This last requirement, if not immediately satisfied, requires some
form of phase compensation at the receiver.}

In particular, our main contributions are the following: i) We provide the
necessary and sufficient conditions ensuring local or global convergence,
for \textit{any} set of finite propagation delays and network topology; ii)
We prove that the convergence is exponential, with convergence rate
depending, in general, on the channel parameters and propagation delays;
iii) We derive a closed form expression for the final consensus, as
a function of the attenuation coefficients and propagation delays of
each link; iv) We show how to get a final, {\it unbiased} function
of the sensor's measurements, which coincides with the globally
optimal decision statistics that it would have been computed by a
fusion center having error-free access to all the nodes. The paper
is organized as follows. Section \ref{Problem formulation} describes
the proposed first-order linearly coupled dynamical system and shows
how to design the system's parameters and the local processing so
that the state derivative of each node converges, asymptotically, to
the globally optimal decision statistics. Section
\ref{Sec:Sync-with-delays} contains the main results of the paper,
namely the necessary and sufficient conditions ensuring the global
or local convergence of the proposed dynamical system, in the
presence of inhomogeneous propagation delays and asymmetric
channels. Finally, Section \ref{Sec:Numerical-Results} contains
numerical results
validating our theoretical findings and draws some conclusions.%

\section{Reaching Globally Optimal Decisions Through Self-Synchronization}

\label{Problem formulation}
It was recently shown in \cite{Giridhar-Kumar05} that, in many
applications, an efficient sensor network design should incorporate
some sort of in-network processing. In this paper, we show first a
class of functions that can be computed with a totally distributed
approach. Then, we illustrate the distributed mechanism able to
achieve the globally optimal decision tests.

\subsection{Forms of Consensus Achievable With a Decentralized Approach}

\label{SUb-Sec_consensus_examples} If we denote by $y_{i}$, $i=1,\ldots ,N$,
the (scalar) measurement taken from node $i$, in a network composed of $N$
nodes, we have shown in \cite{Barbarossa-iwwan05, Barbarossa-Scutari-Journal}
that it is possible to compute any function of the collected data
expressible in the form
\begin{equation}
f(y_{1},y_{2},\ldots ,y_{N})=h\left( \frac{\dsum_{i=1}^{N}c_{i}g_{i}(y_{i})}{%
\dsum_{i=1}^{N}c_{i}}\right) ,  \label{f}
\end{equation}%
where $\left\{ c_{i}\right\} $ are positive coefficients and $g_{i}(\cdot )$%
, $i=1,\ldots ,N,$ and $h(\cdot )$ are arbitrary (possibly nonlinear) real
functions on $\mathbb{R}$, i.e., $g_{i},h:%
\mathbb{R}
\mapsto
\mathbb{R}
$, in a totally decentralized way, i.e., without the need of a sink node. In
the vector observation case, the function may be generalized to the vector
form
\begin{equation}
\mathbf{f}(\mathbf{y}_{1}\mathbf{,y}_{2}\mathbf{,\ldots ,y}_{N})=\mathbf{h}%
\left( \left( \sum_{i=1}^{N}\mathbf{C}_{i}\right) ^{-1}\left( \sum_{i=1}^{N}%
\mathbf{C}_{i}\boldsymbol{g}_{i}\mathbf{(y}_{i}\mathbf{)}\right) \right) ,
\label{f_vect}
\end{equation}%
where $\boldsymbol{y}_{i}=\{y_{i,k}\}_{k=1}^{L}$ is the vector containing
the observations $\{y_{i,k}\}_{k=1}^{L}$ taken be sensor $i,$ and $%
\boldsymbol{g}_{i}(\cdot )$ and $\mathbf{h}(\cdot )$ are arbitrary (possibly
nonlinear) real functions on $\mathbb{R}^{L}$, i.e., $\boldsymbol{g}_{i},$ $%
\mathbf{h}:%
\mathbb{R}
^{L}\mapsto
\mathbb{R}
$, and $\left\{ \mathbf{C}_{i}\right\} $ are arbitrary square positive
definite matrices.

Even though the class of functions expressible in the form (\ref{f}) or (\ref%
{f_vect}) is not the most general one, nevertheless it includes many cases
of practical interest, as shown in the following examples.\newline
\vspace{-0.3cm}

\noindent \textbf{Example 1: ML or BLUE estimate.}\textit{\ } Let us
consider the case where each sensor observes a vector in the form
\begin{equation}
\mathbf{y}_{i}=\mathbf{A}_{i}\mbox{\boldmath
$\xi$}+\mathbf{w}_{i},\qquad i=1,\ldots ,N,  \label{Linmod}
\end{equation}%
where $\mathbf{y}_{i}$ is the $M\times 1$ observation vector, $%
\mbox{\boldmath $\xi$}$ is the $L\times 1$ unknown common parameter vector, $%
\mathbf{A}_{i}$ is the $M\times L$ mixing matrix of sensor $i$, and $\mathbf{%
w}_{i}$ is the observation noise vector, modeled as a circularly symmetric
Gaussian vector with zero mean and covariance matrix $\mathbf{R}_{i}$. We
assume that the noise vectors affecting different sensors are statistically
independent of each other, and that each matrix $\mbox{\boldmath $A$}_{i}$
is full column rank, which implies $M\geq L$. As well known, in this case
the globally optimal 
ML estimate of $\mbox{\boldmath
$\xi$}$ is \cite{Kay-book}:
\begin{equation}
\mathbf{\hat{\xi}}_{ML}=\mathbf{f}(\mathbf{y}_{1}\mathbf{,y}_{2}\mathbf{%
,\ldots ,y}_{N})=\left( \sum_{i=1}^{N}\mathbf{A}_{i}^{T}\mathbf{R}_{i}^{-1}%
\mathbf{A}_{i}\right) ^{-1}\left( \sum_{i=1}^{N}\mathbf{A}_{i}^{T}\mathbf{R}%
_{i}^{-1}\mathbf{y}_{i}\right) .  \label{ML_estimate}
\end{equation}%
This expression is a special case of (\ref{f_vect}), with $\mathbf{C}_{i}=%
\mathbf{A}_{i}^{T}\mathbf{R}_{i}^{-1}\mathbf{A}_{i}$ and $\boldsymbol{g}_{i}(%
\mathbf{y}_{i})=(\mathbf{A}_{i}^{T}\mathbf{R}_{i}^{-1}\mathbf{A}_{i})^{-1}%
\mathbf{A}_{i}^{T}\mathbf{R}_{i}^{-1}\mathbf{y}_{i}$. If the noise pdf is
unknown, (\ref{ML_estimate}) still represents a meaningful estimator, as it
is the BLUE 
\cite{Kay-book}.\newline
\vspace{-0.3cm}

\noindent \textbf{Example 2: Detection of a Gaussian process with unknown
variance embedded in Gaussian noise with known variance. } Let us consider
now a detection problem. Let $y_{i}[k]$ denote the signal observed by sensor
$i$, at time $k$. The detection problem can be cast as a binary hypothesis
test, where the two hypotheses are
\begin{equation}
\begin{array}{l}
\mathcal{H}_{0}:\quad y_{i}[k]=w_{i}[k], \\
\mathcal{H}_{1}:\quad y_{i}[k]=s_{i}[k]+w_{i}[k],%
\end{array}%
\quad i=1,\ldots ,N,\ \ \text{\ }k=1,\ldots ,K,  \label{Binary-det}
\end{equation}%
where $s_{i}[k]$ is the useful signal and $w_{i}[k]$ is the additive noise.
Let us consider the case where the random sequences $s_{i}[k]$ and $w_{i}[k]$
are spatially uncorrelated and modeled as zero mean independent Gaussian
random processes. A meaningful model consists in assuming that the noise
variance is known, let us say equal to $\sigma _{w}^{2}$, whereas the useful
signal variance is not. Under these assumptions, the optimal detector
consists in computing the generalized likelihood ratio test (GLRT) \cite%
{Kay-book, Pescosolido-Barbarossa}
\begin{equation}
T(\mathbf{y})=f(\mathbf{y}_{1},\mathbf{y}_{2},\ldots ,\mathbf{y}_{N})=\frac{1%
}{K}\sum_{i=1}^{N}\sum_{k=1}^{K}y_{i}^{2}[k]\left( \frac{1}{\sigma _{w}^{2}}-%
\frac{1}{\widehat{P}_{i}+\sigma _{w}^{2}}\right) -\sum_{i=1}^{N}\log \left(
\frac{\widehat{P}_{i}+\sigma _{w}^{2}}{\sigma _{w}^{2}}\right) ,
\label{GLRT}
\end{equation}%
and comparing it with a threshold that depends on the desired false alarm
rate. In (\ref{GLRT}), the term $\widehat{P}_{i}$ denotes the ML estimate of
the signal power at node $i$, given by%
\begin{equation*}
\widehat{P}_{i}=\left( \dfrac{1}{K}\dsum_{k=1}^{K}y_{i}^{2}[k]-\sigma
_{w}^{2}\right) ^{+},
\end{equation*}%
where $\left( x\right) ^{+}\triangleq \max \left( 0,x\right) .$ Also in this
case, it is easy to check that (\ref{GLRT}) is a special case of (\ref{f}),
with $c_{i}=1$ and $g_{i}(\mbox{\boldmath $y$}_{i})=\frac{1}{K}%
\sum_{k=1}^{K}y_{i}^{2}[k]\left( \frac{1}{\sigma _{w}^{2}}-\frac{1}{\widehat{%
P}_{i}+\sigma _{w}^{2}}\right) -\log \left( \frac{\widehat{P}_{i}+\sigma
_{w}^{2}}{\sigma _{w}^{2}}\right) $ \cite{Pescosolido-Barbarossa}.

These are only two examples, but the expressions (\ref{f}) or
(\ref{f_vect}) can be used to compute or approximate more general
functions of the collected data, like, maximum and minimum
\cite{Corts-Automatica-06}, geometric mean \cite{Olfati-Franco},
power mean \cite{Bauso-Giarre} and so on, through appropriate choice
of the parameters involved.

\subsection{How to Achieve the Consensus in a Decentralized Way\label%
{System-model}}

The next question is how to achieve the aforementioned
optimal decision statistics with a network having no fusion center. In \cite%
{Barbarossa-iwwan05, Barbarossa-Scutari-Journal} we proposed an
approach to solve this problem using a nonlinear interaction model
among the nodes, based on an undirected graph and with no
propagation delays in the exchange of information among the sensors.
In this paper, we consider a linear interaction model, but we
generalize the approach to a network where the propagation delays
are taken into account and the network is described by a weighted
\emph{directed} graph (or \emph{digraph,} for short), which is a
fairly general model to capture the non reciprocity of the
communication links governing the interaction among the nodes. 

The proposed sensor network is composed of $N$ nodes, each equipped
with four basic components: i) a \textit{transducer} that senses the
physical parameter of interest (e.g., temperature, concentration of
contaminants, radiation, etc.); ii) a \textit{local processing unit}
that processes the measurement taken by the node; iii) a
\textit{dynamical system}, whose state is initialized with the local
measurements and it evolves interactively with the states of nearby
sensors; iv) a \textit{radio interface} that transmits the state of
the dynamical system and receives the state transmitted by the other
nodes, thus ensuring the interaction among the nodes\footnote{The
state value is transmitted by modulating an appropriate carrier.
Because of space limitations, it is not possible to go into the
details of this aspect in this paper. In parallel works, we have
shown that a pulse position modulation of ultrawideband signals may
be a valid candidate for implementing the interaction mechanism.
Some preliminary remarks on the implementation of the proposed
protocol can be found in \cite{Barbarossa-Scutari-Magazine}.}.

\vspace{0.2cm}\noindent \textbf{Scalar observations}\textit{.} In the scalar
observation case, the dynamical system present in node $i$ evolves according
to the following linear functional differential equation
\begin{equation}
\begin{array}{l}
\dot{{x}}_{i}(t;\mathbf{y})=g_{i}(y_{i})+\dfrac{K}{c_{i}}\dsum\limits_{j\in
\mathcal{N}_{i}}a_{ij}\,\left( x_{j}(t-\tau _{ij};\mathbf{y})-x_{i}(t;%
\mathbf{y})\right) ,\quad \,\,t>0, \\
x_{i}(t;\mathbf{y})=\widetilde{\phi }_{i}(t),\quad t\in \lbrack -\tau ,\ 0],%
\end{array}%
\quad
\begin{array}{l}
i=1,\ldots ,N,%
\end{array}
\label{linear delayed system}
\end{equation}%
where $x_{i}(t;\mathbf{y})$ is the state function associated to the $i$-th
sensor that depends on the set of measurements $\mathbf{y}=\{y_{i}\}%
_{i=1}^{N}$; $g_{i}(y_{i})$ is a function of the local observation
$y_{i},$ whose form depends on the specific decision test; $K$ is a
positive coefficient measuring the global coupling strength; $c_{i}$
is a positive coefficient that may be adjusted to achieve the
desired consensus; $\tau _{ij}=T_{ij}+d_{ij}/c$ is a delay
incorporating the propagation delay due to traveling the internode
distance $d_{ij}$, at the speed of light $c$, plus a possible time
offset $T_{ij}$ between nodes $i$ and $j$. The sensors are assumed
to be fixed so that all the delays are constant. We also assume,
realistically, that the maximum delay is bounded, with maximum value
$\tau =\max_{ij}\tau _{ij}.$ The coefficients $a_{ij}$ are assumed
to be nonnegative and, in general, dependent on transmit powers and
channel parameters. For example,  $a_{ij}$ may represent the
amplitude of the signal received from node $i$ and transmitted
from node $j$. In such a case, we have $a_{ij}=\sqrt{P_{j}|h_{ij}|^{2}/d_{ij}^{\eta }}$%
, where $P_{j}$ is the power of the signal transmitted from node $j$; $h_{ij}
$ is a fading coefficient describing the channel between nodes $j$ and $i$; $%
\eta $ is the path loss exponent. The nonnegativity of $a_{ij}$ requires
some form of channel compensation at the receiver side, like in the maximal
ratio receiver, for example, if the channel coefficients $\{h_{ij}\}_{ij}$
are complex variables. Furthermore, we assume, realistically, that node $i$
\textquotedblleft hears\textquotedblright\ node $j$ only if the power
received from $i$ exceeds a given threshold. In such a case, $a_{ij}\neq 0$,
otherwise $a_{ij}=0$. The set of nodes that sensor $i$ hears is denoted by $%
\mathcal{N}_{i}=\{j=1,\ldots ,N:a_{ij}\neq 0\}.$ Observe that, in general, $%
a_{ij}\neq a_{ji}$, i.e., the channels are asymmetric. It is worth
noticing that the state function of, let us say, node $i$ depends,
directly, only on the measurement $y_{i}$ taken by the node itself
and only indirectly on the measurements gathered by the other nodes.
In other words, even though the state $x_{i}(t;\mathbf{y})$ gets to
depend, eventually, on all the measurements, through the interaction
with the other nodes, each node needs to know only its own
measurement.

Because of the delays, the state evolution (\ref{linear delayed system})
for, let us say, $t>0$, is uniquely defined provided that the initial state
variables $x_{i}(t;\mathbf{y})$ are specified in the interval from $-\tau $ to $0$. The
initial conditions of (\ref{linear delayed system}) are assumed to be taken
in the set of continuously differentiable and bounded functions $\widetilde{%
\phi }_{i}(t)$ mapping the interval $[-\tau ,\ 0]$ to $%
\mathbb{R}
$ (see Appendix \ref{Appendix_overview} for more details).

\vspace{0.2cm}\noindent \textbf{Vector observations}. If each sensor
measures a set of, let us say, $L$ physical parameters , the coupling
mechanism (\ref{linear delayed system}) generalizes into the following
expression\footnote{%
We assume that the coupling coefficients $a_{ij}$ are the same for
all estimated parameters. This assumption is justified by the fact
that all the parameters are sent through the same physical channel.}%
\begin{equation}
\begin{array}{l}
\dot{\mathbf{x}}_{i}(t;\mathbf{y})=\mathbf{g}_{i}(\mathbf{y}_{i})+K\ \mathbf{%
Q}_{i}^{-1}\dsum\limits_{j\in \mathcal{N}_{i}}a_{ij}\,\left( \mathbf{x}%
_{j}(t-\tau _{ij};\mathbf{y})-\mathbf{x}_{i}(t;\mathbf{y})\right) ,\quad
\,t>0, \\
\mathbf{x}_{i}(t;\mathbf{y})=\widetilde{\mathbf{\phi }}_{i}(t),\quad t\in
\lbrack -\tau ,\ 0],%
\end{array}%
\quad
\begin{array}{l}
i=1,\ldots ,N,%
\end{array}
\label{vec_dyn_systems_}
\end{equation}%
where $\ \mathbf{x}_{i}(t)$ is the $L$-size vector state of the $i$-th node;
$\boldsymbol{g}_{i}(\mathbf{y}_{i})$ is a vector function of the local
observation $\boldsymbol{y}_{i}=\{y_{i,k}\}_{k=1}^{L}$, i.e. $\boldsymbol{g}%
_{i}:\mathbb{R}^{L}\mapsto \mathbb{R}^{L}$; and $\mathbf{Q}_{i}$ is an $%
L\times L$ non-singular matrix that is a free parameter to be chosen
according to the specific purpose of the sensor network.

As in (\ref{linear delayed system}), the initial conditions of (\ref%
{vec_dyn_systems_}) are assumed to be arbitrarily taken in the set of
continuously differentiable and bounded (vectorial) functions $\widetilde{%
\mathbf{\phi }}_{i}(t)$ mapping the interval $[-\tau ,\ 0]$ to $%
\mathbb{R}
^{L},$ w.l.o.g..

\subsection{Self-Synchronization}

Differently from most papers dealing with average consensus problems \cite%
{Borkar-Varaiya}$-$\cite{Ren-Beard-Control-Magazine} and \cite%
{Schizas-Ribeiro-Giannakis}$-$\cite{Wang-Xiao-06},
where the global consensus was intended to be the situation where all
dynamical systems reach the same state value, we adopt here the alternative
definition already introduced in our previous work \cite%
{Barbarossa-Scutari-Journal}. We define the consensus (through network
synchronization) with respect to the state \textit{derivative}, rather than
to the state, as follows.\vspace{-0.3cm}

\begin{definition}
\label{Definition_sync-state}Given the dynamical system in (\ref{linear
delayed system}) (or (\ref{vec_dyn_systems_})), a solution $\{\mathbf{x}%
_{i}^{\star }(t;\mathbf{y})\}$ of (\ref{linear delayed system}) (or (\ref%
{vec_dyn_systems_})) is said to be a \emph{synchronized state} of the
system, if\vspace{-0.2cm}
\begin{equation}
\dot{\mathbf{x}}_{i}^{\star }(t;\mathbf{y})=\mathbf{\omega }^{\star }(%
\mathbf{y}),\quad \forall i=1,2,\ldots ,N.\vspace{-0.2cm}
\label{Def_synch_state}
\end{equation}%
The system (\ref{linear delayed system}) (or (\ref{vec_dyn_systems_})) is
said to \emph{globally} synchronize if there exists a synchronized state as
in (\ref{Def_synch_state}), and \emph{all} the state derivatives
asymptotically converge to this common value, for \emph{any} given set of
initial conditions $\{\widetilde{\mathbf{\phi }}_{i}\},$ i.e.,
\begin{equation}
\lim_{t\mapsto \infty }\Vert \dot{\mathbf{x}}_{i}(t;\mathbf{y})-\mathbf{%
\omega }^{\star }(\mathbf{y})\Vert =0,\,\,\,\,\,\,\,\forall
\widetilde{\mathbf{\phi }}_{i},\,\,\, i=1,2,\ldots ,N,
\label{def_eq}
\end{equation}%
where $\Vert \cdot \Vert $ denotes some vector norm and $\{\mathbf{x}_{i}(t;%
\mathbf{y})\}$ is a solution of (\ref{linear delayed system}) (or (\ref%
{vec_dyn_systems_})). The synchronized state is said to be \emph{globally
asymptotically stable} if the system globally synchronizes, in the sense
specified in (\ref{def_eq}). The system (\ref{linear delayed system}) (or (%
\ref{vec_dyn_systems_})) is said to \emph{locally} synchronize if there
exist disjoint subsets of the nodes, called clusters, where the nodes in
each cluster have state derivatives converging, asymptotically, to the same
value, for any given set of initial conditions $\{\widetilde{\mathbf{\phi }}%
_{i}\}$. \vspace{-0.3cm}
\end{definition}

Observe that, according to Definition \ref{Definition_sync-state},
if there exists a globally asymptotically stable synchronized state,
then it must necessarily be \emph{unique }(in the derivative). In
the case of local synchronization instead, the system may have
multiple synchronized clusters, each of them with a different
synchronized state. As it will be shown in the next section, one of
the reasons to introduce a novel definition of consensus,
different from the classical one \cite{Borkar-Varaiya}$-$\cite%
{Ren-Beard-Control-Magazine}, \cite{Schizas-Ribeiro-Giannakis}$-$\cite%
{Wang-Xiao-06}, is that the convergence on the state derivative,
rather than on the state, is not affected by the presence of
propagation delays and there is a way to make the final consensus
value to coincide with the globally optimal decision statistic. In
the ensuing sections, we will provide necessary and
sufficient conditions for the system in (\ref{linear delayed system}) (or (%
\ref{vec_dyn_systems_})) to locally/globally synchronize according to
Definition \ref{Definition_sync-state}, along with the closed form
expression of the synchronized state.

\section{Necessary and Sufficient Conditions for Self-Synchronization \label%
{Sec:Sync-with-delays}}

The problem we address now is to check if, in the presence of propagation
delays and asymmetric communication links, the systems (\ref{linear delayed
system})\ and (\ref{vec_dyn_systems_}) can still be used to achieve globally
optimal decision statistics in the form (\ref{f}) and (\ref{f_vect}), in a
totally distributed way. To derive our main results, we rely on some basic
notions of digraph theory. To make the paper self-contained, in Appendix \ref%
{Sec_basic_defs_digraphs} we recall the basic definitions of weak,
quasi-strong and strong connectivity (WC, QSC, and SC, for short) of
a digraph. In Appendix \ref{Sec:Spectral_properties_of_digraph}, we
recall the algebraic properties of the Laplacian matrix associated
to a digraph and we derive the properties of its left eigenvectors,
as they will play a fundamental role in computing the achievable
forms of consensus. In Appendix \ref{Appendix_overview}, we provide
sufficient conditions for the marginal stability of linear delayed
differential equations, as they will be instrumental to prove the
main results of this paper.

The next theorem provides necessary and sufficient conditions for the
proposed decentralized approach to achieve the desired consensus in the
presence of propagation delays and asymmetric communication links.%
\begin{theorem}
\label{Theorem_delay-linear_stability}Let ${\mathscr{G}=}\{{%
\mathscr{V}%
,%
\mathscr{E}}\}$ be the digraph associated to the network in (\ref{linear delayed system}%
), with Laplacian matrix $\mathbf{L}.$ Let $\mathbf{\gamma }=[\gamma
_{1},\ldots ,\gamma _{N}]^{T}$ be a left eigenvector of $%
\mathbf{L}$ corresponding to the zero eigenvalue$,$ i.e., $\mathbf{\gamma }%
^{T}\mathbf{L}=\mathbf{0}_{N}^{T}$. Given the system in (\ref{linear delayed system}),
assume that the following conditions are satisfied:\vspace{-0.4cm}

\begin{description}
\item[a1] The coupling gain $K$ and the coefficients $\left\{ c_{i}\right\}$
are positive; the coefficients $\{a_{ij}\}$ are
non-negative;\vspace{-0.2cm}

\item[a2] The propagation delays $\{\tau _{ij}\}$ are time-invariant and finite, i.e., $\tau
_{ij}\leq \tau =\max_{i\neq j}\tau _{ij}<+\infty ,$ $\forall i\neq j;$%
\vspace{-0.2cm}

\item[a3] The initial conditions are taken in the set of continuously differentiable and bounded
functions mapping the interval $[-\tau ,\ 0]$ to $%
\mathbb{R}
^{N}.$\vspace{-0.2cm}
\end{description}

Then, system (\ref{linear delayed system}) globally synchronizes for
\emph{any} given set of propagation delays, if and only if the
digraph ${\mathscr{G}}$ is
Quasi-Strongly Connected (QSC). The synchronized state is given by%
\begin{equation}
\omega ^{\star
}(\mathbf{y})=\frac{\dsum_{i=1}^{N}\gamma
_{i}c_{i}g_i(y_i)}{\dsum_{i=1}^{N}\gamma
_{i}c_{i}+K\dsum_{i=1}^{N}\dsum_{j\in \mathcal{N}_{i}}\gamma
_{i}a_{ij}\tau _{ij}},
\label{bias_Theo}
\end{equation}%
where $\gamma _{i}>0$ if and only if node $i$ can reach all the
other nodes of the digraph through a strong path, or $\gamma _{i}=0,$
otherwise.

The convergence is exponential, with rate arbitrarily close to $r\triangleq -\min_i\{\left \vert\limfunc{Re}\{s_{i}\}\right \vert: p(s_{i})=0\,\, \text{and}\,\, s_i\neq 0\}$, where $p(s)$ is the characteristic function associated to system (\ref{linear delayed system}) (cf. Appendix \ref{Appendix_overview}).
\end{theorem}

\begin{proof}
See Appendix \ref{proof_Theorem_delay-linear_stability}.
\end{proof}
%
%

Theorem \ref{Theorem_delay-linear_stability} has a very broad applicability,
as it does not make any particular reference to the network topology. If,
conversely, the topology has a specific structure, then we may have the
following forms of consensus.\footnote{%
We focus, w.l.o.g., only on WC digraphs. In the case of non WC digraphs,
Corollary \ref{Corollary} applies to each disjoint component of the digraph.}%
\vspace{-0.2cm}

\begin{corollary}
\label{Corollary}Given system (\ref{linear delayed system}), assume that
conditions \textbf{a1)-a3)} of Theorem \ref{Theorem_delay-linear_stability}
are satisfied. Then,\vspace{-0.3cm}

\begin{enumerate}
\item The system globally synchronizes and the synchronized state is given by%
\begin{equation}
\omega ^{\star }(\mathbf{y})=g_{r}(y_{r}),\quad \quad \,\,r\in \{1,2,\ldots
,N\},  \label{spanning-tree_sync_state}
\end{equation}%
if and only if the digraph ${\mathscr{G}}$ contains only one spanning
directed tree, with root node given by node $r$.\vspace{-0.2cm}

\item The system globally synchronizes and the synchronized state is given
by (\ref{bias_Theo}) with all $\gamma _{i}$'s positive if and only if the
digraph ${\mathscr{G}}$ is Strongly Connected (SC). The synchronized state
becomes
\begin{equation}
\omega ^{\star }(\mathbf{y})=\frac{\dsum_{i=1}^{N}c_{i}g_{i}(y_{i})}{%
\dsum_{i=1}^{N}c_{i}+K\dsum_{i=1}^{N}\dsum_{j\in \mathcal{N}_{i}}a_{ij}\tau
_{ij}},  \label{bias_Theo_balanced_G}
\end{equation}%
if and only if, in addition, the digraph ${\mathscr{G}}$ is balanced.

\item The system \emph{locally} synchronizes in $K$ disjoint clusters ${%
\mathscr{C}%
}_{1},\ldots ,{%
\mathscr{C}%
}_{K}\subseteq \{1,\ldots ,N\}$,\footnote{%
In general, the clusters $\mathscr{C}_{1},\ldots ,{\mathscr{C}}_{K}$ are not
a partition of the set of nodes $\{1,\cdots ,N\}$.} with synchronized state
derivatives for each cluster
\begin{equation}
\dot{x}_{q}^{\star }(t;\mathbf{y}_{k})=\frac{\dsum_{i\in {%
\mathscr{C}%
}_{k}}\gamma _{i}c_{i}g_{i}(y_{i})}{\dsum_{i\in {%
\mathscr{C}%
}_{k}}\gamma _{i}c_{i}+K\dsum_{i\in {%
\mathscr{C}%
}_{k}}\dsum_{j\in \mathcal{N}_{i}}\gamma _{i}a_{ij}\tau _{ij}},\quad \quad
\forall q\in {%
\mathscr{C}%
}_{k},\quad k=1,\ldots ,K,  \label{Eq-Corollary-cluster}
\end{equation}%
with $\mathbf{y}_{k}=\{y_{i}\}_{i\in {%
\mathscr{C}%
}_{k}},$ if and only if the digraph ${\mathscr{G}}$ is weakly connected (WC)
and contains a spanning directed forest with $K$ Root Strongly Connected
Components (RSCC).\footnote{%
Please, see Appendix \ref{Appendix_overview-Graph Theory}.1\ for the formal
definition of (root) strongly connected component. }
\end{enumerate}
\end{corollary}

The proof of this corollary is a particular case of Theorem 1, except that
we exploit the structure of the left eigenvector corresponding to the zero
eigenvalue of the Laplacian matrix $\mathbf{L}$, as derived in Appendix \ref%
{Appendix_overview-Graph Theory}.2.
The previous results can be extended to the vector case in (\ref%
{vec_dyn_systems_}), according to the following.\footnote{%
The proof is omitted because of space limitations, but it follows the same
guidelines as the proof for the scalar case, since the interaction among the
nodes is the same.}

\begin{theorem}
\label{Theorem_delay-linear_stability_vect}Given the system (\ref%
{vec_dyn_systems_}), assume that conditions \textbf{a2)-a3)} of Theorem \ref%
{Theorem_delay-linear_stability} are satisfied and that the matrices $\{%
\mathbf{Q}_{i}\}$ are positive definite. Then, the system synchronizes for
\emph{any} given set of propagation delays, if and only if the digraph ${%
\mathscr{G}}$ is QSC. The synchronized state is given by
\begin{equation}
\mathbf{\omega }^{\star }(\mathbf{y})\triangleq \left( \sum_{i=1}^{N}\gamma
_{i}\mathbf{Q}_{i}+\mathbf{I}_{L}\otimes \left(
K\dsum\limits_{i=1}^{N}\dsum\limits_{j\in \mathcal{N}_{i}}^{N}\gamma
_{i}a_{ij}\tau _{ij}\right) \right) ^{-1}\left( \sum_{i=1}^{N}\gamma _{i}%
\mathbf{Q}_{i}\mathbf{g}_{i}(\mathbf{y}_{i})\right) ,  \label{bias_Theo_vec}
\end{equation}%
where $\otimes $ denotes the Kronecker product and $\gamma _{i}>0$ if and
only if node $i$ can reach all the other nodes of the digraph by a strong
path, or $\gamma _{i}=0,$ otherwise. The convergence is exponential, with
rate arbitrarily close to $r\triangleq -\min_{i}\{\left\vert \limfunc{Re}%
\{s_{i}\}\right\vert :p(s_{i})=0\,\,\text{and}\,\,s_{i}\neq 0\}$, where $p(s)
$ is the characteristic function associated to system (\ref{vec_dyn_systems_}%
).
\end{theorem}


\subsection{\textbf{Impact of propagation delays on convergence}}

The impact of delays in consensus-achieving algorithms has been analyzed in
a series of works \cite[Ch.7.3]{Tsitsiklis-Bertsekas-book}, \cite%
{Olfati-Saber}, \cite{Blondel-Tsitsiklis}, \cite{Strogatz}$-$\cite%
{Wang-Xiao-06}. Among these works, it is useful to distinguish between
consensus algorithms, \cite{Olfati-Saber, Olfati-Saber-Murray-ProcIEEE07},
where the states of all the sensors converge to a prescribed function
(typically the average) of the sensors' initial values, and agreement
algorithms, \cite[Ch.7.3]{Tsitsiklis-Bertsekas-book}, \cite%
{Blondel-Tsitsiklis}, \cite{Strogatz}$-$\cite{Wang-Xiao-06},
\cite{Chellaboina-Haddad-CDC06}, where the goal is to make all the
states to converge to a common value, but without specifying how
this value has to be related to the initial values. In our
application, we can only rely on consensus algorithms, where the
final consensus has to coincide with the globally optimal decision
statistic.

The consensus algorithms analyzed in \cite{Olfati-Saber,
Olfati-Saber-Murray-ProcIEEE07} assume the same delay value for all
the links, i.e., $\tau_{ij}=\tau$, and symmetric channels, i.e., $a_{ij}=a_{ji}$%
, for all $i\neq j$. Under these assumptions, the average consensus in \cite%
{Olfati-Saber, Olfati-Saber-Murray-ProcIEEE07} is reached if and only if the
common delay $\tau $ is upper bounded by \cite{Olfati-Saber}
\begin{equation}
\tau <\frac{\pi }{2}\frac{1}{\lambda _{N}},
\end{equation}%
where $\lambda _{N}$ denotes the maximum eigenvalue of the Laplacian
associated to the undirected graph of the network (cf. Appendix \ref%
{Appendix_overview-Graph Theory}). Since, for any graph with $N$ nodes, we
have \cite{Fielder}
\begin{equation}
\frac{N}{N-1}\ d_{\max }\leq \lambda _{N}\leq 2\ d_{\max },
\end{equation}%
with $d_{\max }$ denoting the maximum graph degree, increasing $d_{\max }$
imposes a strong constraint on the maximum (common) tolerable delay. This
implies, for example, that networks with hubs (i.e., nodes with very large
degrees) that are commonly encountered in scale-free networks \cite{Barabasi}%
, are fragile against propagation delays, when using the consensus
algorithms of \cite{Olfati-Saber, Olfati-Saber-Murray-ProcIEEE07}.

In our application, we were motivated to extend the approach of \cite%
{Olfati-Saber, Olfati-Saber-Murray-ProcIEEE07} to the general case of
inhomogeneous delays and asymmetric channel links. Nevertheless, in spite of
the less restrictive assumptions, Theorem \ref%
{Theorem_delay-linear_stability} shows that our proposed algorithm
is more robust against propagation delays, since its convergence
capability \textit{is not affected by the delays}. Moreover, our
approach is valid in the more general case of asymmetric
communications links and the final value is not simply the average
of the measurements, but a weighted average of functions of the
measurements that can be made to coincide with the desired globally
optimal decision statistic in the form (\ref{f}) or (\ref{f_vect}),
through a proper
choice of the coefficients $c_{i}$. 

An intuitive reason for explaining the main advantage of our
approach is related to the use of an alternative definition of
global consensus: As opposed to conventional methods requiring
consensus on the state value, i.e. \cite{Borkar-Varaiya}$-$\cite%
{Ren-Beard-Control-Magazine}, \cite{Strogatz}$-$\cite{Wang-Xiao-06},
we require the convergence over the state derivative, i.e., we only
require that the state trajectories converge towards parallel
straight lines. The slope must be the same for all the trajectories
and it has to coincide with the desired decision statistic. But the
constant terms of each line may differ from sensor to sensor. This
provides additional degrees of freedom that, eventually, make our
approach more robust against propagation delays or link
coefficients.

\vspace{-0.3cm}

\subsection{Effect of network topology on final consensus value}

Theorem \ref{Theorem_delay-linear_stability} generalizes all the previous
(only sufficient) conditions known in the literature \cite{Olfati-Saber},
\cite{Strogatz}$-$ \cite{Wang-Xiao-06} for the convergence of linear
agreement/consensus protocols in the presence of propagation delays, since
it provides a complete characterization of the synchronization capability of
the system for any possible degree of connectivity in the network (not only
for SC digraphs), as detailed next.

In general, the digraph modeling the interaction among the nodes may
have one of the following structures: i) the digraph contains only
one spanning directed tree, with a single root node, i.e., there
exists only one node that can reach all the other nodes in the
network through a strong directed path; ii) the digraph contains
more than one spanning directed tree, i.e., there exist multiple
nodes (possibly all the nodes), strongly connected to each other,
that can reach all the other nodes by a strong directed path; iii)
the digraph is weakly connected and contains a spanning forest,
i.e., there exists no node that can reach every other node through a
strong directed path; iv) the digraph is not even weakly connected.
The last case is the least interesting, as it corresponds to a set
of isolated subnetworks, that can be analyzed independently of each
other, as one of the previous cases. In the first two cases,
according to Theorem \ref{Theorem_delay-linear_stability}, system
(\ref{linear delayed system}) achieves a \emph{global} consensus,
whereas in the third case the system forms clusters of consensus
with, in general, different consensus values in each cluster, i.e.,
the system synchronizes only \textit{locally}.

In other words, a global consensus is possible \textit{if and only
if} there exists at least one node (the root node of the spanning
directed tree of the digraph) that can send its information,
directly or indirectly, to all the other nodes. If no such a node
exists, a global consensus cannot be reached. However, a
\emph{local} consensus is still achievable among all the nodes that
are able to influence each other.

Interestingly, the closed form expressions of the synchronized state given by (\ref{bias_Theo})
confirms the above statements: The observation $y_{i}$ of, let us
say, sensor $i$ affects the final consensus value if and only if
such an information can reach all the other nodes by a strong
directed path. As a by-product of this result, we have the following
special cases (Corollary 1): If there is only one node that can
reach all the others, then the final consensus depends only on the
observation taken from that node; on the other extreme, the final
consensus contains contributions from \emph{all} the nodes {\it if
and only if} the digraph is SC.

Moreover, if a node contributes to the final value, it does that
through a weight that depends on its in/out-degree. The set of weights $%
\{\gamma _{i}\}$ in (\ref{bias_Theo}) can be interpreted as a measure of the
\textquotedblleft symmetry/asymmetry\textquotedblright\ of the communication
links in the network: Some of these weights are equal to each other if and
only if the subdigraph associated to the corresponding nodes is balanced (or
undirected).

The synchronized state, as given in (\ref{bias_Theo}), suggests also an
interesting interpretation of the consensus formation mechanism of system (%
\ref{linear delayed system}), based on the so called \textit{condensation}
digraph.\footnote{%
Please, refer to Appendix \ref{Sec_basic_defs_digraphs} for the
definition of condensation digraph and the procedure for reducing a
digraph into its condensation digraph.} From (\ref{bias_Theo}) in
fact, one infers that all the nodes that are SC to each other
(usually referred to as nodes of a strongly connected component
(SCC)) produce the same effect on the final consensus as an {\it
equivalent} single node that represents the consensus within that
SCC. In fact, one may easily check if system in (\ref{linear delayed
system}) locally or globally synchronizes and which nodes contribute
on the consensus, simply reducing the original digraph to its
equivalent \emph{condensation }digraph and looking for the existence
of a spanning directed tree in the condensation digraph. The only
SCCs of the original digraph that will provide a contribution on the
final consensus are the SCCs associated to the root nodes of the
condensation digraph.

As an additional remark, the possibility to form clusters of consensus,
rather than a global consensus, depends on the (channel) coefficients $a_{ij}$%
.  If, e.g., these coefficients have the expression as given in
Section \ref{System-model}, they may be altered by changing the
transmit powers $P_j$. As a consequence, the nodes with the highest
transmit power will be the most influential ones. If, for example,
we want to set a certain parameter on each node, like, e.g., a
decision threshold, we can use the same consensus mechanism used in
this paper by assigning, for example, the desired value to, let us
say, node $i$, and select the transmit powers so
that node $i$ is the only node that can reach every other node.
\vspace{-0.2cm}

\subsection{How to get unbiased estimates}

The closed form expression of the synchronized state, as given in (\ref%
{bias_Theo}) (or in (\ref{bias_Theo_vec})), is valid for \emph{any} given
digraph (not only for undirected graphs as in \cite{Olfati-Saber}).
Expression (\ref{bias_Theo}) shows a dependence of the final consensus on
the network topology and propagation parameters, through the coefficients $%
\{a_{ij}\}$, $\{\gamma _{i}\}$ and the delays $\{\tau _{ij}\}$. This means
that, even if the propagation delays do not affect the convergence of the
proposed system, they introduce a bias on the final value, whose amount
depends on both the delays and coefficients $\{a_{ij}\}$. The effect of
propagation parameters and network topology on the final synchronized state
is also contained in the eigenvector $\mathbf{\gamma }$ of the Laplacian $%
\mathbf{L}$. This implies that the final consensus resulting from (\ref%
{bias_Theo}) cannot be made to coincide with the desired decision
statistics as given by (\ref{f}), except that in the trivial case
where all the delays are equal to zero and the digraph is balanced
(and thus strongly connected). However, expression (\ref{bias_Theo})
suggests a method to get rid of any bias, as detailed next.

The bias due to the propagation delays can be removed using the following
two-step algorithm.\footnote{%
We focus only on the scalar system in (\ref{linear delayed system}), because
of space limitation.} We let the system in (\ref{linear delayed system}) to
evolve twice: The first time, the system evolves according to (\ref{linear
delayed system}) and we denote by $\omega^\star(\mbox{\boldmath $y$})$ the
synchronized state; the second time, we set $g_i(y_i)=1$ in (\ref{linear
delayed system}), for all $i$, and the system is let to evolve again,
denoting the final synchronized state by $\omega^\star(\mbox{\boldmath $1$})$%
. From (\ref{bias_Theo}), if we take the ratio $\omega^\star(%
\mbox{\boldmath
$y$})/\omega^\star(\mbox{\boldmath $1$})$, we get
\begin{equation}
\frac{\omega^\star(\mbox{\boldmath $y$})}{\omega^\star(\mbox{\boldmath $1$})}%
=\frac{\dsum_{i=1}^{N}\gamma _{i}c_{i} g_{i}(y_i)}{\dsum_{i=1}^{N}\gamma
_{i}c_{i}},  \label{ratio}
\end{equation}%
which coincides with the ideal value achievable in the absence of delays.
Thus, this simple double-step algorithm allows us to remove the bias term
depending on the delays and on the channel coefficients, without requiring
the knowledge or estimate of neither set of parameters.

If the network is strongly connected and balanced, $\gamma _{i}=1,\forall i$
and then the compensated consensus coincides with the desired value (\ref{f}%
). If the network is unbalanced, the compensated consensus $\omega ^{\star }(%
\mbox{\boldmath $y$})/\omega ^{\star }(\mbox{\boldmath $1$})$ does
not depend on the (channel) coefficients $\{a_{ij}\}$, but it is
still biased, with a bias dependent on $\mathbf{\gamma }$, i.e., on
the network topology. Nevertheless, this residual bias can be
eliminated in a decentralized way according to the following
iterative algorithm. Let us denote by $N_r$ the number of nodes in
the RSCC of the digraph. At the beginning, every node sets
$g_{i}(y_{i})=1$ and $c_{i}=1,$ $i=1, \ldots, N$ and the network is
let to evolve. The final consensus value is denoted by $\omega ^{\star }(%
\mbox{\boldmath $1$})$. Then, the network is let to evolve $N_r$
times, according to the following protocol. At step $i$, with $i=1,
\ldots, N_r$, the nodes within the $i$-th SCC, set $g_{i}(y_{i})=1$,
while all the other nodes set $g_{k}(y_{k})=0$ for all $k\neq i$.
Let us denote by $\omega ^{\star
}(\mbox{\boldmath $e$}_{i})$ the final consensus value, where $%
\mbox{\boldmath $e$}_{i}$ is the canonical vector having all zeros,
except the $i$-th component, equal to one. Repeating this procedure
for all the SCC's, at the end of the $N_{r}$ steps, each node is now
able to compute the ratio $\omega ^{\star }(\mbox{\boldmath
$e$}_{i})/\omega ^{\star }(\mbox{\boldmath $1$})$, which coincides
with $\tilde{\gamma}_{i}:=\gamma _{i}/\sum_{k}\gamma _{k}$. Thus,
after $N_{r}+1$ steps, every node knows its own (normalized)
$\tilde{\gamma}_{i}$. This value is subsequently used to compensate
for the network unbalance as follows. The compensation is achieved
by simply setting, at each node $c_{i}=c_{i}/\tilde{\gamma}_{i}$. In
fact, with this setting, the final ratio $\omega ^{\star
}(\mbox{\boldmath $y$})/\omega ^{\star }(\mbox{\boldmath $1$}) $
coincides with the desired unbiased expression given by (\ref{f}),
for all $i$ such that $\gamma _{i}\neq 0$. If the digraph
$\mathscr{G}$ is SC, this procedure corresponds indeed to make the Laplacian matrix $\mathbf{L}(%
\mathscr{G})$ balanced, so that the left eigenvector of $\mathbf{L}$
associated to the zero eigenvalue be proportional to the vector $\mathbf{1}%
_{N}$. This procedure only needs some kind of coordination among the
nodes to make them working according to a described scheduling. It
is important to remark that, since the eigenvector $\mathbf{\gamma
}$ does not depend on the observations $\{y_{i}\} $, the proposed
algorithm is required to be performed at the start-up phase of the
network and repeated only if the network topology or the channels
change through time. In summary, we can eliminate the effect of both
delays and channel parameters on the final consensus value, thus
achieving the optimal decision statistics as given in (\ref{f}),
with a totally decentralized algorithm, at the price of a slight
increase of complexity and the need for some coordination among the
nodes.\footnote{In \cite{Chellaboina-Haddad-CDC06}, the authors
provided a closed form expression of the agreement obtainable with
an SC balanced network achieving consensus on the state. However,
that expression depends on the channel parameters and on the initial
conditions in a way that the
bias cannot be eliminated through a distributed procedure.}
\subsection{Asymptotic convergence rate}
Theorem \ref{Theorem_delay-linear_stability} generalizes previous results on
the convergence speed of classical linear consensus protocols (see, e.g.
\cite{Olfati-Saber, Olfati-Saber-Murray-ProcIEEE07}) to the case in which
there are propagation delays. In spite of the presence of delays, the
proposed system still converges to the consensus with \emph{exponential}
convergence rate, i.e., $\left\Vert \dot{\mathbf{x}}(t;\mathbf{y})-\mathbf{%
\omega }^{\star }(\mathbf{y})\right\Vert \leq O\left( e^{rt}\right) $, where
the convergence factor $r<0$ is defined in Theorem \ref%
{Theorem_delay-linear_stability}. Thus, the convergence speed is dominated
by the slowest \textquotedblleft mode" of the system. Moreover, as expected,
the presence of delays affects the convergence speed, as the roots of
characteristic equation $p(s)=0$ associated to system (\ref{linear delayed
system}) depend, in general, on both network topology and propagation
delays. Unfortunately, a closed form expression of this dependence is not
easy to derive, and the roots of $p(s)=0$ need to be computed numerically.

In the special case of negligible delays instead, we can provide a
bound of the convergence factor as a function of the channel
parameters through the eigenvalues of the Laplacian matrix
$\mathbf{L}$. Setting $\tau_{ij}=0$ in (\ref{linear
delayed system}), the characteristic equation associated to (%
\ref{linear delayed system}) becomes $p(s)=\left| s\mathbf{I}+\mathbf{L}%
\right| =0$, whose solutions are just the eigenvalues of $-\mathbf{L}$.
Under the assumptions of Theorem \ref{Theorem_delay-linear_stability}, it
follows that
\begin{equation}
r=-\min_i\{\limfunc{Re}\{\lambda_{i}\}: \lambda_i \in \sigma\{\mathbf{L}\},
\text{and}\,\, \lambda_i\neq 0\},  \label{characteristic-roots-nodelays}
\end{equation}
where $\sigma\{\mathbf{L}\}$ denotes the spectrum of Laplacian
matrix $\mathbf{L}$. The value of $r$ is negative if and only if the
digraph $\mathscr{G}$ is QSC
\cite[Lemma 2]{Wu-Linear_Algebra_2}. Using (\ref%
{characteristic-roots-nodelays}), we can obtain bounds on the convergence
rate as a function of the network topology, as shown next. When the digraph
is strongly connected we have\footnote{%
The inequality in (\ref{lambda_2}) follows from the fact that $\limfunc{Re}%
\{\lambda_{i}\}\geq |\kappa|$ for all nonzero eigenvalues $\{\lambda_{i}\}$
of $\mathbf{L}$ \cite{Wu-Linear_Algebra_2}.}
\begin{equation}
r\leq \kappa \triangleq -\lambda_2\left(\frac{1}{2}(\mathbf{D}_{\gamma}%
\mathbf{L}+\mathbf{L}^T\mathbf{D}_{\gamma})\right)<0,  \label{lambda_2}
\end{equation}
where $\mathbf{D}_{\gamma}\triangleq \limfunc{diag}(\gamma_1,\cdots,
\gamma_N)$, $\mathbf{\gamma}$ is the left eigenvector of $\mathbf{L}$
associated to the zero eigenvalue, normalized so that $\| \mathbf{\gamma}%
\|_{\infty}=1$, and $\lambda_2(\mathbf{A})$ denotes the second smallest
eigenvalue of the symmetric matrix $\mathbf{A}$. It follows from (\ref%
{lambda_2}) that system (\ref{linear delayed system}) reaches a consensus
with rate at least $\kappa$. As special cases of (\ref{lambda_2}) we have
the following: i) In the case of balanced digraphs, $\mathbf{D}_{\gamma}=%
\mathbf{I}$ and thus $\kappa=-\lambda_2\left(\frac{1}{2}(\mathbf{L}+\mathbf{L%
}^T)\right)$; ii) If the digraph is undirected (and connected) $\mathbf{D}%
_{\gamma}=\mathbf{I}$, $\mathbf{L}=\mathbf{L}^T$ and thus $%
\kappa=-\lambda_2\left(\mathbf{L}\right)$, where $\lambda_2\left(\mathbf{L}%
\right)$ is also known as the \emph{algebraic connectivity} of the digraph
\cite{Fielder}. For the case in which the digraph is QSC, some bounds of $r$
can be found in \cite{Wu-Linear_Algebra, Wu-synchronization} using the
generalization of the classical definition of algebraic connectivity.
Moreover, interestingly, the convergence rate of the system, under the
assumption of Theorem \ref{Theorem_delay-linear_stability}, can be related
to the convergence rates of the SCCs of condensation digraph associated to
the system. Because of lack of space, we suggest the interested reader to
check \cite{Wu-Linear_Algebra_2} for more details.

In conclusion, according to the above results, we infer that the
convergence rate of the proposed consensus algorithm, at least in
the absence of propagation delays, is the same as that of the
classical linear protocols achieving consensus on the state.

\section{Numerical Results}

\label{Sec:Numerical-Results} In this section, we illustrate first some
examples of consensus, for different network topologies. Then, we show an
application of the proposed technique to an estimation problem, in the
presence of random link coefficients. In both examples, the analog system (%
\ref{linear delayed system}) is implemented in discrete time, with sampling
step size $T_s=10^{-3}$.\medskip

\noindent \textbf{Example 1: Different forms of consensus for different
topologies}\newline
In Figure \ref{Figure-RSCC}, we consider three topologies (top row), namely:
(a) a SC digraph, (b) a QSC digraph with three SCCs, and (c) a WC (not QSC)
digraph with a spanning forest composed by two trees. For each digraph, we
also sketch its decomposition into SCCs (each one enclosed in a circle),
corresponding to the nodes of the associated condensation digraph (whose
root SCC is denoted by RSCC). In the bottom row of Figure \ref{Figure-RSCC},
we plot the dynamical evolutions of the state derivatives of system (\ref%
{linear delayed system}) versus time, for the three network topologies,
together with the theoretical asymptotic values predicted by (\ref{bias_Theo}%
) (dashed line with arrows). As proved by Theorem \ref%
{Theorem_delay-linear_stability}, the dynamical system in Figure \ref%
{Figure-RSCC}a) achieves a global consensus, since the underlying
digraph is SC. The network of Figure \ref{Figure-RSCC}b), instead,
is not SC, but the system is still able to globally synchronize,
since there is a set of nodes, in the RSCC component, able to reach
all the other nodes. The final consensus, in such a case, contains
only the contributions of the nodes in the RSCC, since no other node
belongs to the root of a spanning directed tree of the condensation
digraph. Finally, the system in Figure \ref{Figure-RSCC}c) cannot
reach a global consensus since there is no node that can reach all
the others, but it does admit two disjoint clusters, corresponding
to the
two RSCCs, namely RSCC$_{1}$ and RSCC$_{2}$. The middle lines of Figure \ref%
{Figure-RSCC}c) refer to the nodes of the SCC component, not belonging to
either RSCC$_{1}$ or RSCC$_{2}$, that are affected by the consensus achieved
in the two RSCC components, but that cannot affect them. Observe that, in
all the cases, the state derivatives of the (global or local) clusters
converge to the values predicted by the closed form expression given in (\ref%
{bias_Theo}), (\ref{spanning-tree_sync_state}) or (\ref{Eq-Corollary-cluster}%
), depending on the network topology.
\begin{figure}[tbh]
\hspace{-0.9cm}\includegraphics[height=10 cm]{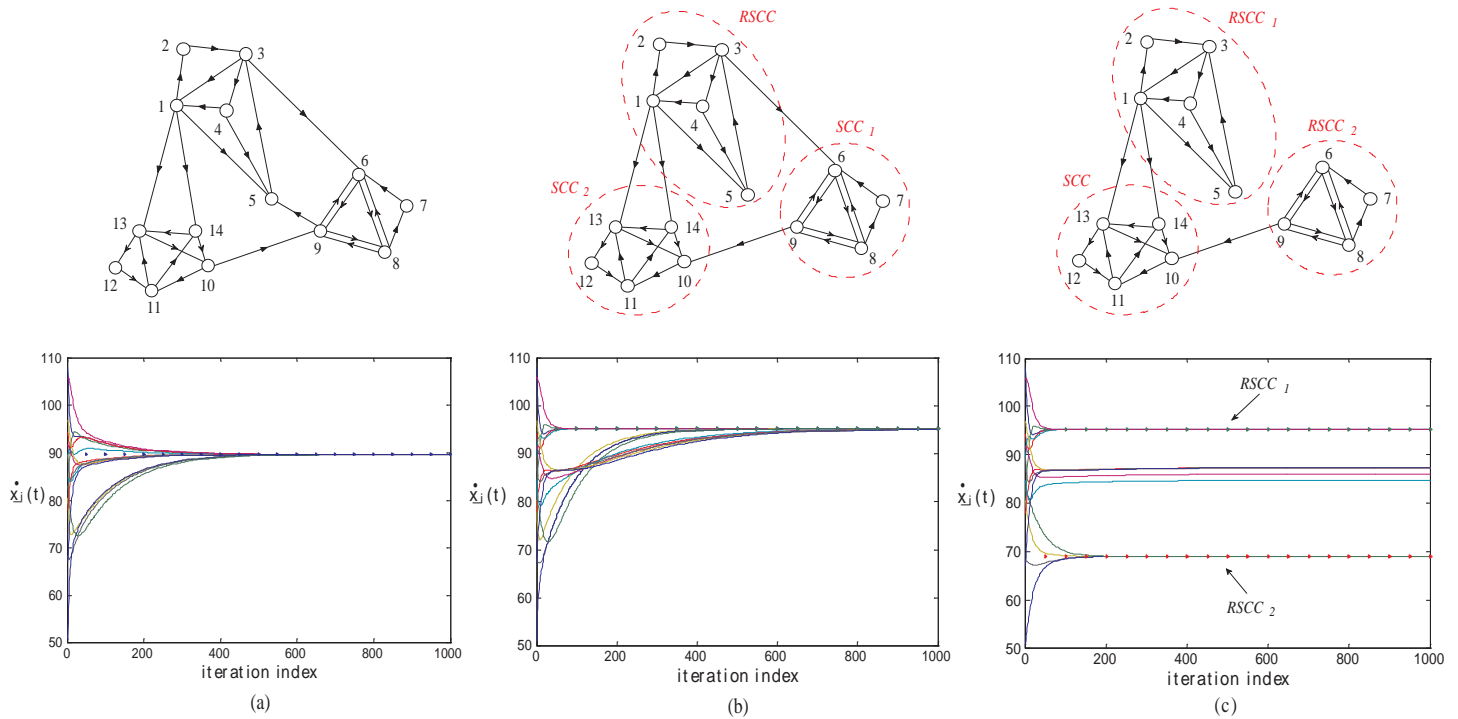}
\caption{{\protect\small Consensus for three different network topologies:
a) SC digraph; b) QSC digraph with three SCCs; c) WC digraph with a two
trees forest; $T_{s}=10^{-3}$, $\protect\tau =50T_{s}$, $K=30$, N=14. }}
\label{Figure-RSCC}
\end{figure}\medskip


\noindent \textbf{Example 2: Distributed optimal decisions through
self-synchronization}. The behaviors shown in the previous example refer to
a given realization of the topology, with given link coefficients, and of
the observations. In this example, we report a global parameter representing
the variance obtained in the estimate of a scalar variable. Each sensor
observes a variable $y_{i}=A_{i}\xi +w_{i}$, where $w_{i}$ is additive zero
mean Gaussian noise, with variance $\sigma _{i}^{2}$. The goal is to
estimate $\xi $. The estimate is performed through the interaction system (%
\ref{linear delayed system}), with functions $g_{i}(y_{i})=y_{i}/A_{i}$ and
coefficients $c_{i}=A_{i}^{2}/\sigma _{i}^{2}$, chosen in order to achieve
the globally optimal ML estimate. The network is composed of $40$ nodes,
randomly spaced over a square of size $D$. The size of the square occupied
by the network is chosen in order to have a maximum delay $\tau =100T_{s}$.
We set the threshold on the amplitude of the minimum useful signal to zero,
so that, at least in principle, each node hears each other node. The
corresponding digraph is then SC. To simulate a practical scenario, the
channel coefficients $a_{ij}$ are generated as i.i.d. Rayleigh random
variables, to accommodate for channel fading. Each variable $a_{ij}$ has a
variance depending on the distance $d_{ij}$ between nodes $i$ and $j$, equal
to\footnote{%
We use the attenuation factor $1/(1+d_{ij}^{2})$ instead of $1/d_{ij}^{2}$
to avoid the undesired event that, for $d_{ij}<1$ the received power might
be greater than the transmitted power.} $\sigma
_{ij}^{2}=P_{j}/(1+d_{ij}^{2})$.

In Figure \ref{Figure-variance-vs-time}, we plot the estimated average state
derivative (plus and minus the estimation standard deviation), as a function
of the iteration index. Figure \ref{Figure-variance-vs-time}a) refers to the
case in which there is only observation noise, but there is no noise in the
exchange of information among the nodes. Conversely, Figure \ref%
{Figure-variance-vs-time}b) refers to the case where the node
interaction is noisy, so that the state evolution of each sensor is
described by the state equation
$\dot{z}_{i}(t)=\dot{x}_{i}(t)+v_{i}(t)$, with $\dot{x}_{i}(t)$
given by (\ref{linear delayed system}), where $v_{i}(t)$ is white
Gaussian noise; the SNR is $|\xi |^{2}/\sigma _{w}^{2}=20$ dB. The
averages are taken across all the nodes, for $100$ independent
realizations of the network, where, in each realization we generated
a new topology and a new set of channel coefficients and noise
terms. The results refer to the following cases of interest: a) ML
estimate achieved with a centralized system, with no communication
errors between nodes and fusion center (dotted lines); b) estimate
achieved with the proposed method, with no propagation delays, as a
benchmark term (dashed and dotted lines plus $\times $ marks for the
average value); c) estimate achieved with the proposed method, in
the presence of propagation delays (dashed lines plus triangles for
the average value); d) estimate achieved with the two-step
estimation method leading to (\ref{ratio}) (solid lines plus circles
for the average value).
\begin{figure}[tbp]
\centering\vspace{-0.4cm}
\subfigure[\,] {\includegraphics[height=6cm,
width=7.6cm]{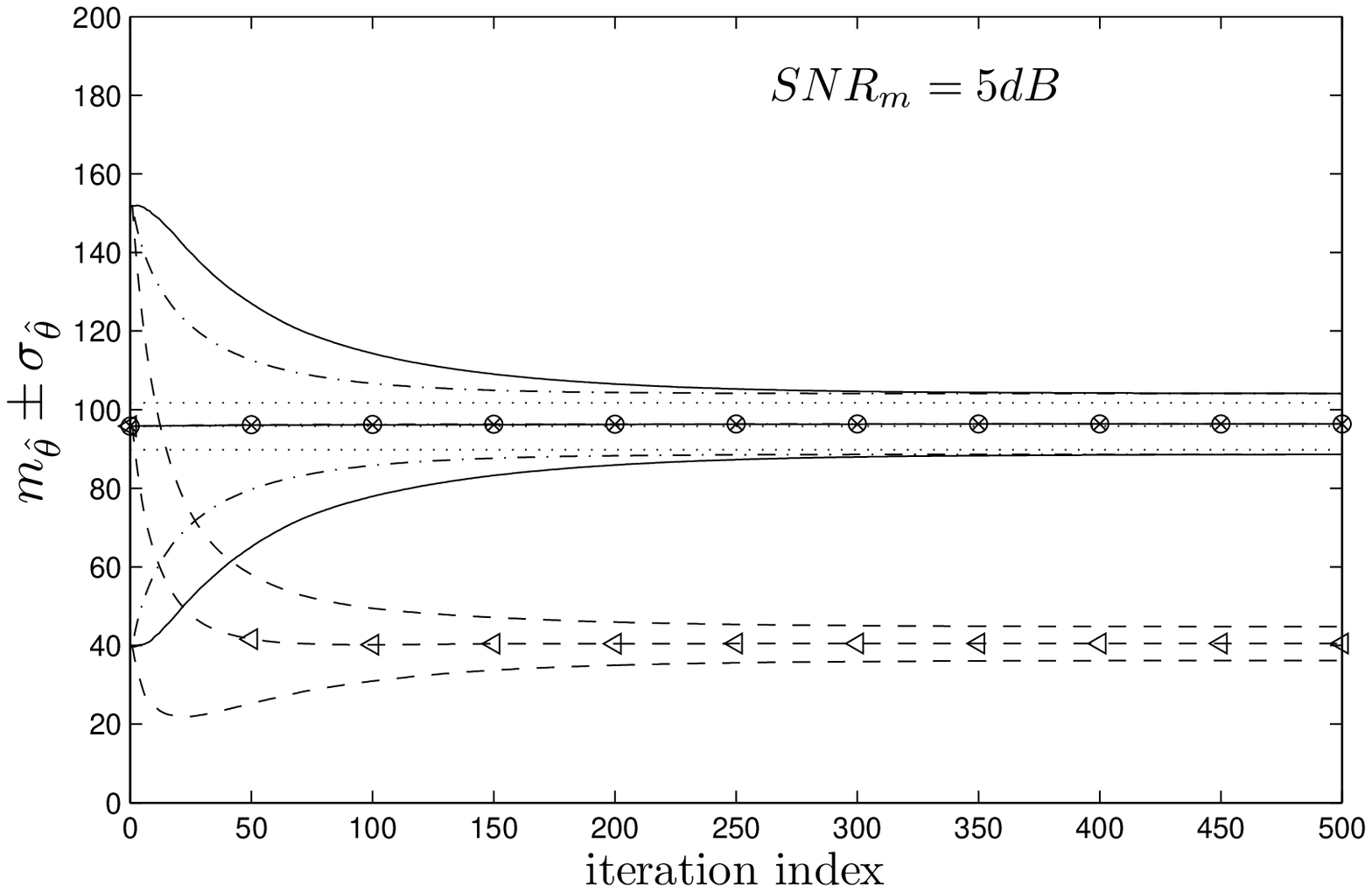} \hfill}\hspace{0.2cm}
\subfigure[]{\includegraphics[height=6cm,
width=7.9cm]{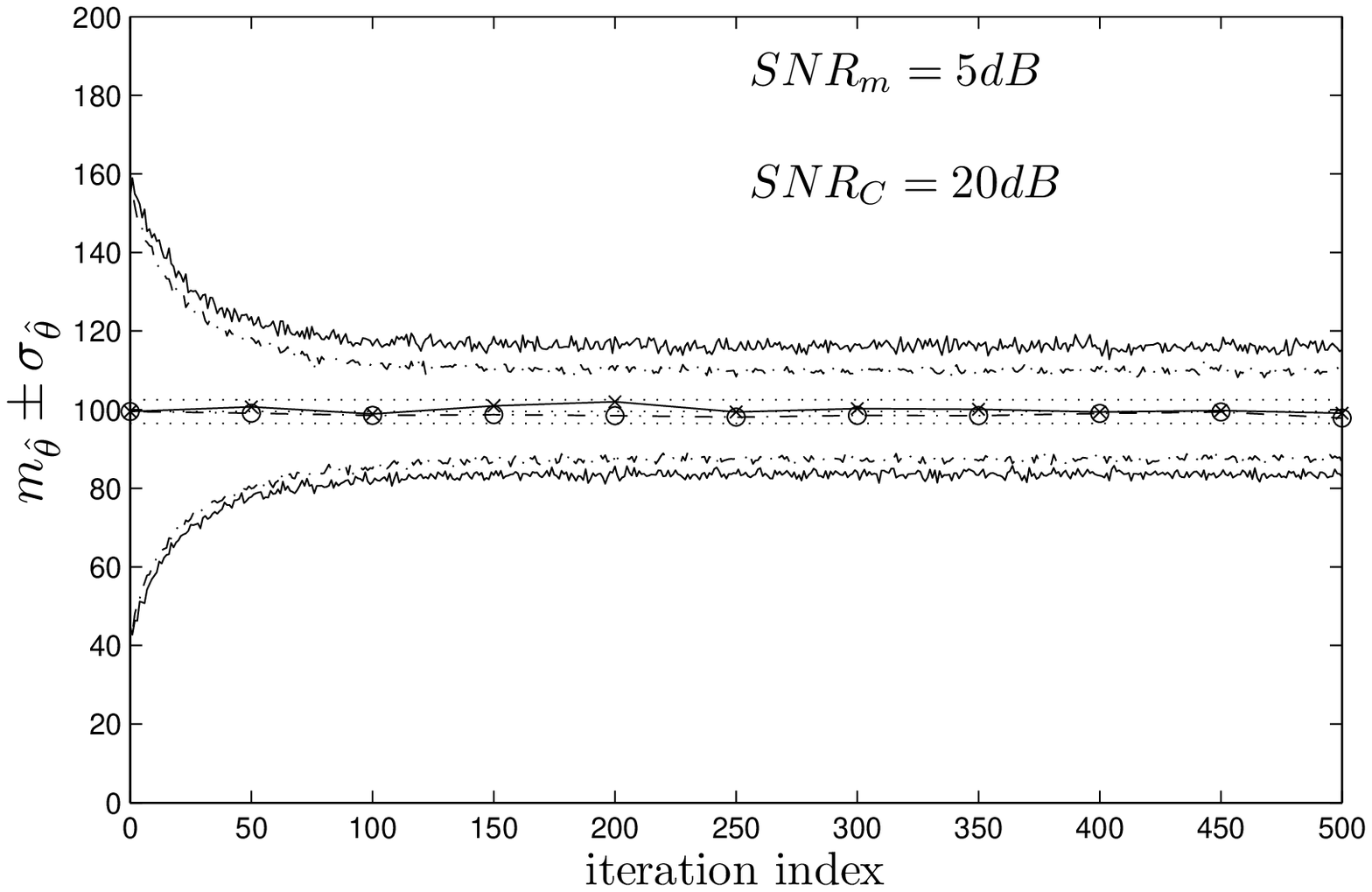}
\hfill}
\vspace{-0.5cm}
\caption{{\protect\small Estimated parameter vs. convergence time in the
absence [subplot a)] and in the presence [subplot b)] of receiver noise: a)
ML estimate by a centralized system (dotted lines); b) estimate by the
proposed method, with no propagation delays, (dashed and dotted lines plus
\textquotedblleft $\times $\textquotedblright for the average value); c)
estimate by the proposed method, with propagation delays (dashed lines plus
\textquotedblleft $\lhd $\textquotedblright \thinspace\ for the average
value); d) estimate by the two-step estimation method leading to (\protect
\ref{ratio}) (solid lines plus \textquotedblleft $\circ $\textquotedblright
\thinspace\ for the average value). }}
\label{Figure-variance-vs-time}
\end{figure}
From Figure \ref{Figure-variance-vs-time} we can see that, in the
absence of delays, the (decentralized) iterative algorithm based on
(\ref{linear delayed system}) behaves, asymptotically, as the
(centralized) globally optimal ML estimator. In the presence of
delays, we observe a clear bias (dashed lines), due to the large
delay values, but with a final estimation variance still close to
the ML estimator's. Interestingly, the two-step procedure leading to
(\ref{ratio}) provides results very close to the optimal ML
estimator, with no apparent bias, in spite of the large delays and
the random channel fading coefficients. The only price paid with the
two-step procedure, besides time, is a slight increase of the
variance due to taking the ratio of two noisy consensus values, as evidenced in Fig. \ref%
{Figure-variance-vs-time} b).

\section{Concluding remarks}

In conclusion, in this paper we have proposed a totally
decentralized sensor network scheme capable to reach globally
optimal decision tests through local exchange of information among
the nodes, in the presence of asymmetric communication channels and
inhomogeneous time-invariant propagation delays. The method is
particularly useful for applications where the goal of the network
is to take decisions about a common event. Differently from the
average consensus protocols available in the literature, our system
globally synchronizes for \textit{any} set of (finite) propagation
delays \emph{if and only if} the underlying digraph is QSC, with a
final synchronized state that is a known function of the sensor
measurements. In general, the synchronized state depends on the
propagation parameters, such as delays and communication channels.
Nevertheless, exploiting our closed form expression for the final
consensus values, we have shown how to recover an unbiased estimate,
for any set of delays and channel coefficients, without the need to
knowing or estimating these coefficients. This desirable result is a
distinctive property of the consensus achievable on the state
derivative and cannot be obtained using classical
consensus/agreement algorithms that reach a consensus on the state
variables. If we couple the nice properties
mentioned above with the properties reported in \cite%
{Barbarossa-Scutari-Journal}, where we showed that, in the absence of
delays, the consensus protocol proposed in this paper and in \cite%
{Barbarossa-Scutari-Journal} is also robust against coupling noise, we have,
overall, a good candidate for a distributed sensor network.

As in many engineering problems, the advantages of our scheme come with
their own prices. Three issues that deserve further investigations are the
following: i) the states grow linearly with time; ii) the coefficients $%
a_{ij}$ are nonnegative; and iii) a change of topology affects the
convergence properties of the proposed scheme. The first issue has an impact
on the choice of the radio interface responsible for the exchange of
information between the nodes. To avoid the need for transmitting with a
high dynamic range, the nodes must transmit a nonlinear, bounded function of
the state value. One possibility, as proposed in \cite%
{Barbarossa-Scutari-Magazine}, is to associate the state value to the phase
of sinusoidal carrier or to the time shift of a pulse oscillator. The second
point requires that the receiver be able to compensate for possible sign
inversions. As far as the switching topology is concerned, it would be
useful to devise methods aimed at increasing the resilience of our method
against topology changes. But it is worth keeping in mind that the previous
aspects are only the reverse of the medal of a method capable to achieve a
globally optimal decision for any set of delays and for asymmetric channels.

\section{Appendix}

\appendix\vspace{-0.3cm}

\section{Directed Graphs}

\label{Appendix_overview-Graph Theory} The interaction among the sensors is
properly described by a directed graph. For the reader's convenience, in
this section, we briefly review the notation and basic results of graph
theory that will be used throughout this paper. For the reader interested in
a more in-depth study of this field, we recommend, for example, \cite%
{Fielder}$-$\cite{Godsil-Royle book}.\vspace{-0.2cm}

\subsection{Basic Definitions\label{Sec_basic_defs_digraphs}}

To take explicitly into account the possibility of unidirectional links
among the network nodes, we represent the information topology among the
nodes by their (weighted) directed graph.

\noindent \textbf{Directed graph}. Given $N$ nodes, a (weighted) directed
graph (or digraph) ${\mathscr{G}}$ is defined as ${\mathscr{G}=}\{{%
\mathscr{V}%
,%
\mathscr{E}%
}\},$ where ${%
\mathscr{V}%
}\triangleq \{v_{1},\ldots ,v_{N}\}$ is the set of nodes (or vertices) and ${%
\mathscr{E}%
}\subseteq {%
\mathscr{V}%
\times
\mathscr{V}%
}$ is a set of edges (i.e., \emph{ordered} pairs of the nodes), with the
convention that $e_{ij}\triangleq (v_{i},v_{j})\in $ ${%
\mathscr{E}%
}$ (i.e., $v_{i}$ and $v_{j}$ are the head and the tail of the edge $e_{ij},$
respectively) means that the information flows from $v_{j}$ to $v_{i}.$ A
digraph is weighted if a positive weight is associate to each edge,
according to a proper map ${%
\mathscr{W}%
:%
\mathscr{E}%
\rightarrow
\mathbb{R}
}_{+},{\ }${such that }if $e_{ij}\triangleq (v_{i},v_{j})\in $ ${%
\mathscr{E}%
,}$ {then} ${%
\mathscr{W}%
(e_{ij})=a}_{ij}>0,$ otherwise ${a}_{ij}=0.$ We focus in the following on
weighted digraphs where the weights of loops $(v_{i},v_{i})$ are zero, i.e.,
$a_{ii}=0$ for all $i.$ \ If $(v_{i},v_{j})\in $ ${%
\mathscr{E}%
\Leftrightarrow }(v_{j},v_{i})\in $ ${%
\mathscr{E}%
}$ (and $a_{ij}=a_{ji},$ $\forall i\neq j$), then the graph is said to be
(weighted) \emph{undirected. }For any node $v_{i}\in {%
\mathscr{V}%
,}$\emph{\ }we define\emph{\ }the information neighbor of $v_{i}$ as
\begin{equation}
\mathcal{N}_{i}\triangleq \left\{ j=1,\ldots ,N:e_{ij}=(v_{i},v_{j})\in {%
\mathscr{E}%
}\right\} .  \label{neighbor_set}
\end{equation}%
The set $\mathcal{N}_{i}$ represents the set of indices of the nodes sending
data to node $i$.

The in-degree and out-degree of node $v_{i}\in {%
\mathscr{V}%
}$ are, respectively, defined as:%
\begin{equation}
\deg \nolimits_{\text{in}}(v_{i})\triangleq \dsum_{j=1}^{N}a_{ij},\quad
\text{and}\quad \deg _{\text{out}}(v_{i})\triangleq \dsum_{j=1}^{N}a_{ji}.
\label{in-out degree}
\end{equation}%
Observe that for undirected graphs, $\deg \nolimits_{\text{in}}(v_{i})=\deg
\nolimits_{\text{out}}(v_{i}).$

We may have the following class of digraphs.

\noindent \textbf{Balanced digraph}\emph{\ }The node $v_{i}$ of a digraph ${%
\mathscr{G}=}\{{%
\mathscr{V}%
,%
\mathscr{E}%
}\}$ is said to be \emph{balanced }if and only if its in-degree and
out-degree coincide, i.e., $\deg \nolimits_{\text{in}}(v_{i})=\deg
\nolimits_{\text{out}}(v_{i}).$ A digraph ${\mathscr{G}=}\{{%
\mathscr{V}%
,%
\mathscr{E}%
}\}$ is called \emph{balanced }if and only if all its nodes are balanced,
i.e.,
\begin{equation}
\dsum_{j=1}^{N}a_{ij}=\dsum_{j=1}^{N}a_{ji},\quad \forall i=1,\ldots ,N.
\label{balanced-cond}
\end{equation}

\noindent \textbf{Path/cycle }A \emph{strong path }(or directed chain) in a
digraph ${\mathscr{G}=}\{{%
\mathscr{V}%
,%
\mathscr{E}%
}\}$ is a sequence of \emph{distinct }nodes $v_{0},\ldots ,v_{q}\in {%
\mathscr{V}%
}$ such that $(v_{i},v_{i-1})\in $ ${%
\mathscr{E}%
,}$ $\forall i=1,\ldots ,q.$ If $v_{0}\equiv v_{q},$ the path is said to be
\emph{closed. }A \emph{weak }path is a sequence of \emph{distinct }nodes $%
v_{0},\ldots ,v_{q}\in {%
\mathscr{V}%
}$ such that either $(v_{i-1},v_{i})\in $ ${%
\mathscr{E}%
}$ or $(v_{i},v_{i-1})\in $ ${%
\mathscr{E}%
,}$ $\forall i=1,\ldots ,q.$ A\ (strong) \emph{cycle }(or circuit) is a
closed (strong) path.

\noindent \textbf{Directed tree/forest}\emph{\ } A digraph with $N$ nodes is
a (rooted) \emph{directed tree }if it has $N-1$ edges and there exists a
distinguished node, called the \emph{root }node,\emph{\ }which can reach
\emph{all} the other nodes by a (unique) strong path. Thus a directed tree
cannot have cycles and every node, except the root, has one and only one
incoming edge.\footnote{%
Observe that some literature (e.g., \cite{Lin-Francis, Jin-Murray}) defines
this concept using the opposite convention for the orientation of the edges.}
A digraph is a (directed) \emph{forest }if it consists of one or more
directed trees. A subgraph ${\mathscr{G}}_{s}{=}\{{%
\mathscr{V}%
}_{s}{,%
\mathscr{E}%
}_{s}\}$ of a digraph ${\mathscr{G}=}\{{%
\mathscr{V}%
,%
\mathscr{E}%
}\},$ with ${%
\mathscr{V}%
}_{s}\subseteq {%
\mathscr{V}%
}$ and ${%
\mathscr{E}%
}_{s}\subseteq {%
\mathscr{E}%
\cap }\left( {%
\mathscr{V}%
}_{s}\times {%
\mathscr{V}%
}_{s}\right) ,$ is a \emph{spanning }directed tree (or a \emph{spanning }%
directed\emph{\ }forest), if it is a directed tree (or a directed forest)
and it has the \emph{same} node set as ${\mathscr{G},}$ i.e., ${%
\mathscr{V}%
}_{s}\equiv {%
\mathscr{V}%
.}$ We say that a digraph ${\mathscr{G}=}\{{%
\mathscr{V}%
,%
\mathscr{E}%
}\}$ contains a spanning tree (or a spanning forest) if there exists a
subgraph of ${\mathscr{G}}$ that is a spanning directed tree (or a spanning
directed forest).

\noindent \textbf{Connectivity}\emph{. }In a digraph there are many
degrees of connectedness. In this paper we focus on the following. A
digraph is \emph{Strongly }Connected (SC) if, for \emph{every} pair
of nodes $v_{i}$ and $v_{j}$, there exists a strong path from
$v_{i}$ to $v_{j}$ and viceversa. A digraph is
\emph{Quasi}-\emph{Strongly }Connected (QSC) if, for \emph{every}
pair of nodes $v_{i}$ and $v_{j}$, there exists a node $r$ that can
reach both $v_{i}$ and $v_{j}$ by a \emph{strong} path. A digraph is
\emph{weakly }connected (WC) if any pair of distinct nodes can be
joined by a \emph{weak} path. A digraph is \emph{disconnected }if it
is not weakly\emph{\ }connected. \

According to the above definitions, it is straightforward to see that strong
connectivity implies quasi strong connectivity and that quasi strong
connectivity implies weak connectivity, but the converse, in general, does
not hold. For undirected graphs instead, the above notions of connectivity
are equivalent: An undirected graph is connected if any two distinct nodes
can be joined by a path. Moreover, it is easy to check that the quasi strong
connectivity of a digraph is equivalent to the existence of at least one
spanning directed tree in the graph (see, e.g., \cite[p. 133 ]%
{Berge-Houri-book}).

\noindent \textbf{Condensation Digraph }When a digraph ${\mathscr{G}}$ is
WC, it may still contain strongly connected subgraphs. A maximal subgraph of
${\mathscr{G}}$ which is also SC is called \emph{Strongly Connected
Component }(SCC) of ${\mathscr{G}}${\footnote{%
Maximal means that there is no larger SC subgraph containing the nodes of
the considered component.} \cite{Brualdi-Ryser-book}.\ Using this concept,
any digraph }${\mathscr{G}=}\{{%
\mathscr{V}%
,%
\mathscr{E}%
}\}$ can be partitioned into SCCs, let us say ${\mathscr{G}}_{1}{\triangleq }%
\{{%
\mathscr{V}%
}_{1}{,%
\mathscr{E}%
}_{1}\},$ $\ldots ,{\mathscr{G}}_{K}{\triangleq }\{{%
\mathscr{V}%
}_{K}{,%
\mathscr{E}%
}_{K}\},$ where ${%
\mathscr{V}%
}_{j}\subseteq {%
\mathscr{V}%
}$ and ${%
\mathscr{E}%
}_{j}\subseteq {%
\mathscr{E}%
}$ denote the set of nodes and edges lying in the $j$-th SCC, respectively.
Using this structure, one can reduce the original digraph ${\mathscr{G}}$ to
the so called \emph{condensation }digraph ${\mathscr{G}}^{\star }{=}\{{%
\mathscr{V}%
}^{\star }{,%
\mathscr{E}%
}^{\star }\}{,}$ {by substituting each node set }${%
\mathscr{V}%
}_{i}$ of {each SCC }${\mathscr{G}}_{i}$ of ${\mathscr{G}}$ with a distinct
node $v_{i}^{\star }\in {%
\mathscr{V}%
}^{\star }$of ${\mathscr{G}}^{\star },$ and introducing an edge in ${%
\mathscr{G}}^{\star }$ from $v_{i}^{\star }$ to $v_{j}^{\star }$ if and only
if there exists \emph{some }edges from the $i$-th SCC ${\mathscr{G}}_{i}$
and the $j$-th SCC ${\mathscr{G}}_{j}$ {\cite[Ch. 3.2]{Brualdi-Ryser-book}.
An SCC that is reduced to the root node of a directed tree of the
condensation digraph is called \emph{Root }}SCC (RSCC).{\ Observe that, by
definition, the condensation digraph has no cycles \cite[Lemma 3.2.3]%
{Brualdi-Ryser-book}}. Building on this property, we have the following.
\vspace{-0.2cm}

\begin{lemma}
\label{condensation-digraph-vertex-ordering}Let ${\mathscr{G}}^{\star }{=}\{{%
\mathscr{V}%
}^{\star }{,%
\mathscr{E}%
}^{\star }\}$ be the condensation digraph of $\mathscr{G},$ composed
by $K$ nodes. Then, the nodes of \ ${\mathscr{G}}^{\star }$ can
always be ordered
as $v_{1}^{\star },\ldots ,v_{K}^{\star }\in {%
\mathscr{V}%
}^{\star },$ so that all the existing edges in ${\mathscr{G}}^{\star }$ are
in the form%
\begin{equation}
(v_{i}^{\star },v_{j}^{\star })\in {%
\mathscr{E}%
}^{\star },\text{ \quad with \quad }1\leq j<i\leq K,
\label{Condensation_digraph_ordering}
\end{equation}%
where $v_{1}^{\star }$ has zero in-degree$.$\vspace{-0.2cm}
\end{lemma}

The ordering $v_{1}^{\star },\ldots ,v_{K}^{\star }$ satisfying (\ref%
{Condensation_digraph_ordering}) can be obtained by the following iterative
procedure. Starting from $v_{1}^{\star },$ remove $v_{1}^{\star }$ and all
its out-coming edges from ${\mathscr{G}}^{\star }.$ Since the reduced
digraph with $K-1$ nodes has no (strong) cycles by construction, there must
exist a node with zero in-degree in it. Let us denote such a node by $%
v_{2}^{\star }.$ Then, no edges in the form $(v_{2}^{\star },v_{j}^{\star
}), $ with $j>2,$ can exist in the reduced digraph (and thus in ${\mathscr{G}%
}^{\star }$). This justifies (\ref{Condensation_digraph_ordering}) for $i=2$
and $j=1,2.$ The rest of (\ref{Condensation_digraph_ordering}), for $j>2$,
is obtained by repeating the same procedure on the remaining nodes.

The connectivity properties of a digraph are related to the structure of its
\emph{condensation }digraph, as given in the following Lemma (we omit the
proof because of space limitations).\vspace{-0.2cm}

\begin{lemma}
\label{condensation-digraph}Let ${\mathscr{G}}^{\star }$ be the condensation%
\emph{\ }digraph of ${\mathscr{G}.}$ {Then: i) }${\mathscr{G}}${\ is SC if
and only if }${\mathscr{G}}^{\star }$is composed by a single node; ii) ${%
\mathscr{G}}${\ is QSC if and only if }${\mathscr{G}}^{\star }$contains a
spanning\ directed tree; iii) if ${\mathscr{G}}${\ is WC, then }${\mathscr{G}%
}^{\star }${\ contains either a spanning directed tree or a (weakly)
connected directed forest}.\vspace{-0.2cm}
\end{lemma}

The concept of condensation digraph is useful to understand the network
synchronization behavior, as shown in Section \ref{Sec:Sync-with-delays}. \
A similar idea was already used in \cite{Chai-Wu-2005, Jin-Murray} to study
leadership problems in coordinated multi-agent systems.\vspace{-0.2cm}

\subsection{Spectral properties of a Digraph\label%
{Sec:Spectral_properties_of_digraph}}

We recall now some basic relationships between the connectivity properties
of the digraph and the spectral properties of the matrices associated to the
digraph, since they play a fundamental role in the stability analysis of the
system proposed in this paper. In the following, we denote by $\mathbf{1}%
_{N} $ and $\mathbf{0}_{N}$ the $N$-length column vector of all ones and
zeros, respectively.

Given a digraph ${\mathscr{G}=}\{{%
\mathscr{V}%
,%
\mathscr{E}%
}\},$ we introduce the following matrices associated to ${\mathscr{G}}$:

\begin{itemize}
\item The $N \times N$ \emph{adjacency matrix} $\mathbf{A}$ is composed of
entries $\left[ \mathbf{A}\right] _{ij}\triangleq a_{ij},$ $i,j=1,\ldots ,N$
, equal to the weight associated with the edge $e_{ij}$, if $e_{ij}\in $ ${%
\mathscr{E}%
,}$ or equal to zero, otherwise;

\item The \emph{degree } \emph{matrix} $%
\boldsymbol{\Delta}%
$ is a diagonal matrix with diagonal entries $\left[
\boldsymbol{\Delta}%
\right] _{ii}\triangleq \deg _{\text{in}}(v_{i})$;

\item The (weighted) \emph{Laplacian } $\mathbf{L}$ is defined as%
\begin{equation}
\left[ \mathbf{L}\right] _{ij}\triangleq \left\{
\begin{array}{ll}
\dsum_{k\neq i=1}^{N}a_{ik}, & \text{if }j=i, \\
-a_{ij} & \text{if }j\neq i.%
\end{array}%
\right.  \label{Laplacian}
\end{equation}%
Using the adjacency matrix $\mathbf{A}$ and the degree matrix $%
\boldsymbol{\Delta}%
$, the Laplacian can be rewritten in compact form as $\mathbf{L}\triangleq
\boldsymbol{\Delta}%
-\mathbf{A.}$\footnote{%
Observe that the definition of Laplacian matrix as given in (\ref{Laplacian}%
) coincides with that used in the classical graph theory literature, except
for the convention we adopted in the orientation of the edges. This leads to
$\mathbf{L}$ expressed in terms of the in-degrees matrix, rather than of the
out-degrees matrix. Our choice is motivated by the physical interpretation
we gave to the edges' weights, as detailed in Section \ref{System-model}.}
\end{itemize}

By definition, the Laplacian matrix $\mathbf{L}$ in (\ref{Laplacian}) has
the following properties: i) it is a diagonally dominant matrix \cite%
{Horn-book}; ii) it has zero row sums; and iii) it has nonnegative diagonal
elements. From i)-iii), invoking Ger\v{s}gorin's disk Theorem \cite%
{Horn-book}, we have that zero is an eigenvalue of $\mathbf{L}$
corresponding to a right eigenvector in the $\limfunc{Null}\{\mathbf{L}%
\}\supseteq \limfunc{span}\{\mathbf{1}_{N}\},$ i.e.,
\begin{equation}
\mathbf{L\ 1}_{N}=\mathbf{0}_{N},  \label{zero-eig}
\end{equation}%
while all the other eigenvalues have positive real part. This also means
that $\limfunc{rank}(\mathbf{L})\leq N-1.$

Moreover, from (\ref{balanced-cond}) and (\ref{zero-eig}), it turns out that
\emph{balanced} digraphs can be equivalently characterized in terms of the
Laplacian matrix $\mathbf{L}$: A digraph is balanced if and only if $\mathbf{%
1}_{N}$ is also a left eigenvector of $\mathbf{L}$ associated with the zero
eigenvalue, i.e.,
\begin{equation}
\mathbf{1}_{N}^{T}\mathbf{L}=\mathbf{0}_{N}^{T},  \label{eq-balanced-cond}
\end{equation}%
or equivalently $\frac{1}{2}(\mathbf{L}+\mathbf{L}^T)$ is positive
semidefinite.

The relationship between the connectivity properties of a digraph and the
spectral properties of its Laplacian matrix are given by the following.%
\vspace{-0.3cm}

\begin{lemma}
\label{Lemma_spanning-forest}Let ${\mathscr{G}=}\{{%
\mathscr{V}%
,%
\mathscr{E}%
}\}$ be a digraph with Laplacian matrix $\mathbf{L}.$ The multiplicity of
the zero eigenvalue of $\mathbf{L}$ is equal to the minimum number of
directed trees contained in a spanning directed forest of ${\mathscr{G}.}$
\end{lemma}

\vspace{-0.4cm}

\begin{corollary}
\label{Lemma_spanning-tree}The zero eigenvalue of $\mathbf{L}$ is \emph{%
simple} if and only if ${\mathscr{G}}$ contains a spanning directed tree
(or, equivalently, it is QSC).
\end{corollary}

Lemma \ref{Lemma_spanning-forest} comes directly form Theorem 9 and Theorem
10 of \cite{Wu-Linear_Algebra}. Corollary \ref{Lemma_spanning-tree} was
independently proved in many papers, such as \cite[Corollary 2]{ren-et-al},
\cite[Lemma 2]{Lin-Francis}. Observe that, since the strong connectivity of
the digraph implies QSC, the results provided in \cite{Olfati-Saber,
Jin-Murray} for SC digraphs, can be obtained as special case of Corollary %
\ref{Lemma_spanning-tree}. Specifically, we have the following.

\begin{corollary}
\label{Lemma_spanning-tree-SC}Let ${\mathscr{G}=}\{{%
\mathscr{V}%
,%
\mathscr{E}%
}\}$ be a digraph with Laplacian matrix $\mathbf{L}.$ If ${\mathscr{G}}$ is
SC, then $\mathbf{L}$ has a \emph{simple} zero eigenvalue and a positive
left-eigenvector associated to the zero eigenvalue.
\end{corollary}

According to Corollary \ref{Lemma_spanning-tree}, because of (\ref{zero-eig}%
), the Laplacian of a QSC digraph has an isolated eigenvalue equal to zero,
corresponding to a right eigenvector in the $\limfunc{span}\{\mathbf{1}%
_{N}\}.$ Observe that, for \emph{undirected }graphs, Corollary \ref%
{Lemma_spanning-tree-SC} can be stated as follows: $\limfunc{rank}(\mathbf{L}%
)=N-1$ if and only if ${\mathscr{G}}$ is connected \cite{Fielder}. For
directed graphs, instead, the \textquotedblleft only if\textquotedblright\
part does not hold.

We derive now the structure of the left-eigenvector $\mathbf{\gamma
}$ of the Laplacian matrix $\mathbf{L}$ associated to the zero
eigenvalue, as a function of the network topology. This result is
instrumental to prove the main theorem of this paper. We have the
following.\vspace{-0.2cm}

\begin{lemma}
\label{Lemma_eigenvector_Laplacian}Let ${\mathscr{G}=}\{{%
\mathscr{V}%
,%
\mathscr{E}%
}\}$ be a digraph with $N$ nodes and Laplacian matrix $\mathbf{L}.$ Assume
that ${\mathscr{G}}$ is QSC with $K$ SCC's ${\mathscr{G}}_{1}{\triangleq }\{{%
\mathscr{V}%
}_{1}{,%
\mathscr{E}%
}_{1}\},$ $\ldots ,{\mathscr{G}}_{K}{\triangleq }\{{%
\mathscr{V}%
}_{K}{,%
\mathscr{E}%
}_{K}\},$ with ${%
\mathscr{V}%
}_{i}\subseteq {%
\mathscr{V}%
,}$ ${%
\mathscr{E}%
}_{i}\subseteq {%
\mathscr{E}%
}$, $\left\vert {%
\mathscr{V}%
}_{i}\right\vert =r_{i}$ and $\sum_{i}r_{i}=N,$ numbered w.l.o.g. so that ${%
\mathscr{G}}_{1}$ coincides with the RSCC of $\mathscr{G}$. Then, the
left-eigenvector $\mathbf{\gamma =[\gamma }_{1},\ldots ,\mathbf{\gamma }_{N}%
\mathbf{]}^{T}$ of $\mathbf{L}$ associated to the zero eigenvalue has the
following structure%
\begin{equation}
\mathbf{\gamma }_{i}=\left\{
\begin{array}{ll}
>0, & \quad \text{iff }v_{i}\in {%
\mathscr{V}%
}_{1}{,} \\
=0, & \quad \text{otherwise.}%
\end{array}%
\right.
\end{equation}%
If ${\mathscr{G}}_{1}$ is also balanced, then $\mathbf{\gamma }%
_{r_{1}}\triangleq \mathbf{[\gamma }_{1},\ldots ,\mathbf{\gamma }_{r_{1}}%
\mathbf{]}^{T}\in \limfunc{span}\{\mathbf{1}_{r_{1}}\},$ where $r_{1}{%
\triangleq }\left\vert {%
\mathscr{V}%
}_{1}\right\vert .$
\end{lemma}

\begin{proof}
Since the digraph ${\mathscr{G}}$ is QSC, the set ${%
\mathscr{V}%
}_{1}$ contains either all $N$ or $0<r_{1}<N$ nodes of ${\mathscr{G}.}$

In the former case, ${%
\mathscr{V}%
}_{1}\equiv {%
\mathscr{V}%
}$ and thus the digraph is SC, by definition. Hence, according to Corollary %
\ref{Lemma_spanning-tree-SC} (cf. Appendix \ref%
{Sec:Spectral_properties_of_digraph}), the Laplacian matrix $\mathbf{L(}%
\mathscr{G}\mathbf{)}$ has a simple zero eigenvalue with left-eigenvector $%
\mathbf{\gamma >0.}$ If, in addition, $\mathscr{G}$ is balanced (and thus
also SC), then $\mathbf{\gamma }\in \limfunc{span}\left\{ \mathbf{1}%
_{N}\right\} .$

We consider now the latter case, i.e., $0<r_{1}<N.$ According to Lemma %
\ref{condensation-digraph}, the condensation digraph ${\mathscr{G}}^{\star }{%
=}\{{%
\mathscr{V}%
}^{\star }{,%
\mathscr{E}%
}^{\star }\}$ of $\mathscr{G}$ contains a spanning directed tree composed by
$K$ nodes $v_{1}^{\star },\ldots ,v_{K}^{\star }\in {%
\mathscr{V}%
}^{\star }$ (associated to the $K$ SCCs ${\mathscr{G}}_{1},$ $\ldots ,{%
\mathscr{G}}_{K})$, with root node $v_{1}^{\star }$ (associated to the SCC ${%
\mathscr{G}}_{1}$ of ${\mathscr{G}}$). From Lemma \ref%
{condensation-digraph-vertex-ordering}, we assume, w.l.o.g., that the nodes $%
v_{1}^{\star },\ldots ,v_{K}^{\star }\in {%
\mathscr{V}%
}^{\star }$ are ordered according to (\ref{Condensation_digraph_ordering}),
so that the Laplacian matrix of ${\mathscr{G}}^{\star }$ be a lower
triangular matrix.

Using the relationship between the original digraph ${\mathscr{G}}$ and its
condensation digraph ${\mathscr{G}}^{\star }$ (cf. Appendix \ref%
{Sec_basic_defs_digraphs}), the Laplacian matrix $\mathbf{L(}\mathscr{G}%
\mathbf{)}$ of $\mathscr{G}$ can be written as a lower \emph{block}
triangular matrix, i.e. 
\begin{equation}
\mathbf{L(}\mathscr{G}\mathbf{)}=\left(
\begin{array}{cccc}
\mathbf{L}_{1} & \mathbf{0} & \ldots & \mathbf{0} \\
\ast & \mathbf{B}_{2} & \ddots & \vdots \\
\ast & \ast & \ddots & \mathbf{0} \\
\ast & \ldots & \ast & \mathbf{B}_{K}%
\end{array}%
\right) ,  \label{L_bar}
\end{equation}%
with $\mathbf{L}_{1}=\mathbf{L}(\mathscr{G}_{1})\in
\mathbb{R}
^{r_{1}\times r_{1}}$ denoting the $r_{1}\times r_{1}$ Laplacian matrix
associated to the root SCC $\mathscr{G}_{1},$ and $\mathbf{B}_{k}=\mathbf{L}%
_{k}+\mathbf{D}_{k}\in
\mathbb{R}
^{r_{k}\times r_{k}},$ where $\mathbf{L}_{k}=\mathbf{L}(\mathscr{G}_{k})$ is
the $r_{k}\times r_{k}$ Laplacian matrix of the $k$-th SCC $\mathscr{G}_{k}$
of $\mathscr{G}$ and $\mathbf{D}_{k}$ is\ a nonnegative diagonal matrix
whose $i$-th entry is equal to the sum of the weights associated to the
edges outgoing from the nodes in $\mathscr{G}_{1},\ldots ,\mathscr{G}_{k-1}$
and incoming in the $i$-th node in $\mathscr{G}_{k}$.\footnote{%
We assume, w.l.o.g., that in each SCC $\mathscr{G}_{k}$ the nodes are
numbered from $1$ to $r_{k}.$} Observe that, since $\mathscr{G}$ is QSC,
each $\mathbf{D}_{k}$ has at least one positive entry; otherwise the SCC $%
\mathscr{G}_{k}$ would be decoupled from the other SCCs.

Since each $\mathscr{G}_{k}$ is SC by definition, we have $\limfunc{rank}(%
\mathbf{L}_{k})$ $=r_{k}-1$ (cf. Corollary \ref{Lemma_spanning-tree-SC}) and
$\limfunc{Null}(\mathbf{L}_{k})=\limfunc{span}(\mathbf{1}_{r_{k}})$ (see (%
\ref{zero-eig})). Using these properties and the fact that $\mathbf{D}_{k}%
\mathbf{1}_{r_{k}}\neq \mathbf{0}_{r_{k}}$, we have that the rows (columns)
of $\mathbf{L}_{k}+\mathbf{D}_{k}$ are linearly independent, or
equivalently, %
%
%
\begin{equation}
\limfunc{rank}(\mathbf{B}_{k})=r_{k},\quad \forall k=2,\ldots ,K.
\label{full_rank_l_k}
\end{equation}%
%
%
%
%
%
%
%
%
%
%
%
%
%
%
%
%
%
%

Using (\ref{L_bar}) and (\ref{full_rank_l_k}), we derive now the structure
of the left-eigenvector $\mathbf{\gamma }$ of $\mathbf{L}(\mathscr{G})$ in (%
\ref{L_bar}). Partitioning $\mathbf{\gamma }$ $\mathbf{=[\gamma }%
_{r_{1}}^{T},\mathbf{\gamma }_{r_{2}}^{T},\ldots ,\mathbf{\gamma }%
_{r_{K}}^{T}\mathbf{]}^{T}$ according to (\ref{L_bar}), with $\mathbf{\gamma
}_{r_{k}}\in
\mathbb{R}
^{r_{k}},$ we have%
\begin{equation}
\mathbf{\gamma }_{r_{1}}^{T}\mathbf{L}_{1}=\mathbf{0}_{r_{1}}^{T}\quad \text{%
and\quad }\mathbf{\gamma }_{r_{k}}^{T}\mathbf{B}_{k}=\mathbf{0}_{r_{k}}^{T},%
\text{ \quad }k=2,\ldots ,K,  \label{left-eig-subgraph}
\end{equation}%
which, using Corollary \ref{Lemma_spanning-tree-SC} and (\ref{full_rank_l_k}%
), provides%
\begin{equation}
\mathbf{\gamma }_{r_{1}}>\mathbf{0}_{r_{1}}\quad \text{and\quad }\mathbf{%
\gamma }_{r_{k}}=\mathbf{0}_{r_{k}},\text{\quad }k=2,\ldots ,K,
\end{equation}%
respectively. If, in addition, $\mathscr{G}_{1}$ is balanced, then $\mathbf{%
\gamma }_{r_{1}}\in \limfunc{span}\left\{ \mathbf{1}_{r_{1}}\right\} .$
\end{proof}

\vspace{-0.3cm}

\section{Preliminary Results on Linear Functional Differential Equations
\label{Appendix_overview}}

In this section, we introduce some basic definitions on linear
functional differential equations and prove some intermediate
results that will be
extensively used in the proof of Theorem \ref%
{Theorem_delay-linear_stability}, as given in Appendix \ref%
{proof_Theorem_delay-linear_stability}.

To formally introduce the concept of functional differential equations, we
need the following notations: let $%
\mathbb{R}
^{N}$ be an $N$-dimensional linear vector space over the real numbers with
vector norm $\left\Vert \cdot \right\Vert ;$ denoting by $\tau $ the maximum
time delay of the system, let $\mathcal{C}\triangleq \mathcal{C}\left(
[-\tau ,\ 0],%
\mathbb{R}
^{N}\right) \mathcal{\ }$[or $\mathcal{C}^{1}\triangleq \mathcal{C}%
^{1}\left( [-\tau ,\ 0],%
\mathbb{R}
^{N}\right) $] be the Banach space of continuous (or continuously
differentiable) functions mapping the interval $[-\tau ,\ 0]$ to $%
\mathbb{R}
^{N}$ with the topology of uniform convergence, i.e., for any $\mathbf{\phi }%
\in $ $\mathcal{C}$ (or $\mathbf{\phi }\in $ $\mathcal{C}^{1}$) the norm of $%
\mathbf{\phi }$ is defined as $\left\vert \mathbf{\phi }\right\vert
_{s}=\sup_{-\tau \leq \vartheta \leq 0}\left\Vert \mathbf{\phi }(\vartheta
)\right\Vert $.

Since in this paper we considered only linear coupling among the
differential equations as given in (\ref{linear delayed system}), we
focus on \emph{Linear homogeneous }Delayed Differential Equations
(LDDEs) \
with a finite number of heterogeneous non-commensurate delays\footnote{%
Observe that the LDDE (\ref{system-translated}) falls in the more general
class of LDDE's studied in \cite{Bellman-Cooke}-\cite{Gu-book}. In fact, it
is straightforward to check that LDDE (\ref{system-translated}) can be
rewritten in the canonical form of \cite[Ch. 6, Eq. (6.3.2)]{Bellman-Cooke},
\cite[Ch. 3, Eq. (3.1)]{Gu-book}.} \cite{Bellman-Cooke}-\cite{Kuang-book}:

\begin{equation}
\overset{\cdot }{x}_{i}(t)=k_{i}\dsum\limits_{j\in \mathcal{N}%
_{i}}a_{ij}\left( x_{j}(t-\tau _{ij})-x_{i}(t)\right) ,\quad \quad \quad
\vspace{-0.2cm}\medskip \smallskip \text{ \ \ }i=1,\ldots ,N,\quad t>t_{0},
\label{system-translated}
\end{equation}
where $k_i$ is any arbitrary positive constant. Equation (\ref%
{system-translated}) means that the derivative of the state variables $%
\mathbf{x}$ at time $t$ depends on $t$ and $\mathbf{x}(\vartheta ),$ for $%
\vartheta \in \lbrack t-\tau ,$ $t].$ Hence, to uniquely specify the
evolution of the state beyond time $t_{0}$, it is required to specify the
initial state variables $\mathbf{x}(t)$ in a time interval of length $\tau ,$
from $t_{0}-\tau $ to $t_{0},$ i.e.,%
\begin{equation}
x_{i}(t_{0}+\vartheta )=\phi _{i}(\vartheta ),\text{ \quad }\vartheta \in
\lbrack -\tau ,\ 0],\text{ \ }i=1,\ldots ,N\text{.}
\end{equation}%
Here after, the initial conditions $\mathbf{\phi =}\left\{ \phi _{i}\right\}
_{i}$ are assumed to be taken in the set $\mathcal{C}_{\beta }^{1}$ of
(continuously differentiable) functions that are bounded in the norm%
\footnote{%
We used, without loss of generality, as vector norm $\left\Vert \cdot
\right\Vert $ in $%
\mathbb{R}
^{N}$ the infinity norm $\left\Vert \cdot \right\Vert _{\infty },$ defined
as $\left\Vert \mathbf{x}\right\Vert _{\infty }\triangleq \max_{i}|x_{i}|.$}
$\left\vert \mathbf{\phi }\right\vert _{s}=\sup_{-\tau \leq \vartheta \leq
0}\left\Vert \mathbf{\phi }(\vartheta )\right\Vert _{\infty }$, i.e.,%
\begin{equation}
\mathcal{C}_{\beta }^{1}\triangleq \left\{ \mathbf{\phi }\in \mathcal{C}%
^{1}:\left\vert \phi \right\vert _{s}=\sup_{-\tau \leq \vartheta \leq
0}\left\Vert \mathbf{\phi }(\vartheta )\right\Vert _{\infty }\leq \beta
<+\infty \text{ }\right\} .  \label{C_beta}
\end{equation}

Given $\mathbf{\phi }\in \mathcal{C}_{\beta }^{1}$ and $t_{0},$ let $\mathbf{%
x}[t_{0},\mathbf{\phi }](t)$ denote the function $\mathbf{x}$\textbf{\ }at
time\textbf{\ }$t,$ with initial value $\mathbf{\phi }$ at time $t_{0},$
i.e., $\mathbf{x}(t_{0}+\vartheta )=\mathbf{\phi }(\vartheta ),$ with $%
\vartheta \in \lbrack -\tau ,\ 0]\mathbf{.}$ For the sake of notation, we
define $\mathbf{x}[\mathbf{\phi }](t)\triangleq \mathbf{x}[0,\mathbf{\phi }%
](t).$ A function $\mathbf{x}[t_{0},\mathbf{\phi }](t)$ is said to be a
solution to equation (\ref{system-translated}) on $[t_{0}-\tau ,\
t_{0}+t_{1})$ with initial value $\mathbf{\phi }$ at time $t_{0},$ if there
exist $t_{1}\geq 0$ such that: i) $\mathbf{x}\in \mathcal{C}\left(
[t_{0}-\tau ,\ t_{0}+t_{1}],%
\mathbb{R}
^{n}\right) ;$ ii) $\mathbf{x}[t_{0},\mathbf{\phi }](t)$ satisfies (\ref%
{system-translated}), $\forall t\in \lbrack t_{0,}$ $t_{0}+t_{1}];$ and iii)
$\mathbf{x}(t_{0}+\vartheta )=\mathbf{\phi }(\vartheta ),$ with $\vartheta
\in \lbrack -\tau ,\ 0].$ It follows from \cite[Theorem 1.2]{Gu-book} that
such a solution to equation (\ref{system-translated}) exists and is unique.
We focus now on two concepts related to the trajectories of (\ref%
{system-translated}), namely boundedness and stability.

\begin{definition}
\label{Def:bounded}Given system (\ref{system-translated}), a solution $%
\mathbf{x}[t_{0},\mathbf{\phi }](t)$ is \emph{bounded }if there exists a $%
\beta =\beta (t_{0},\mathbf{\phi })$ such that $\left\Vert \mathbf{x}[t_{0},%
\mathbf{\phi }](t)\right\Vert <\beta (t_{0},\mathbf{\phi })$ for $t\geq
t_{0}-\tau .$ The solutions are \emph{uniformly bounded }if,\emph{\ }for any%
\emph{\ }$\alpha >0,$ there exists a $\beta =\beta (\alpha )>0$ such that
for all $t_{0}\in
\mathbb{R}
,$ $\mathbf{\phi }\in $ $\mathcal{C}$ and $\left\vert \mathbf{\phi }%
\right\vert _{s}<\alpha $ we have $\left\Vert \mathbf{x}[t_{0},\mathbf{\phi }%
](t)\right\Vert <\beta $ for all $t\geq t_{0}.$
\end{definition}

As far as the stability notion is concerned, it is not at all different from
its counterpart for unretarded systems, except for the different assumptions
on the initial conditions. The interested reader may refer, e.g., to \cite%
{Gu-book, Kuang-book} for an in-depth treatment of this topic. We
focus here on
effective methods on proving the stability of LDDEs. Since system (\ref%
{system-translated}) is linear, the stability analysis can be carried out
either in the time-domain \cite{Gu-book, Kuang-book} or in the
frequency-domain \cite{Bellman-Cooke, Stepan-book}, as for classical
ordinary differential equations (see, e.g., \cite{Kailath-book}). In this
paper we focus on the latter approach. The same conclusions can also be
obtained using the time-domain analysis, based on Lyapunov-Krasovskii
functional \cite{Haddock-Terjeki}. Before stating the major result, we need
first the following intermediate definitions.

Let $\mathbf{%
\mathbb{C}
}_{+}=\{s\in \mathbf{%
\mathbb{C}
}:\limfunc{Re}\{s\}>0\},$ $\mathbf{%
\mathbb{C}
}_{-}=\{s\in \mathbf{%
\mathbb{C}
}:\limfunc{Re}\{s\}<0\},$ and $\overline{\mathbf{%
\mathbb{C}
}}_{+}$ be the closure of $\mathbf{%
\mathbb{C}
}_{+}$, i.e., $\overline{\mathbf{%
\mathbb{C}
}}_{+}=\{s\in \mathbf{%
\mathbb{C}
}:\limfunc{Re}\{s\}\geq 0\}.$ Denoting by $\mathcal{H}^{n\times m}$ the set
of $n\times m$\ matrices whose entries are analytic\footnote{%
A complex function is said to be analytic (or holomorphic) on a region $%
\mathcal{D\subseteq
\mathbb{C}
}$ if it is complex differentiable at every point in $\mathcal{D}$, i.e.,
for any $z_{0}\in \mathcal{D}$ the function satisfies the Cauchy-Riemann
equations and has continuous first partial derivatives in the neighborhood
of $z_{0}$ (see, e.g., \cite[Theorem 11.2]{Rudin-book}).} and bounded
functions in $\mathbf{%
\mathbb{C}
}_{+},$ let us introduce the $N\times N$ diagonal degree matrix $%
\boldsymbol{\Delta}%
\geq \mathbf{0}_{N\times N}$ and the complex matrix $\mathbf{H}(s)\in
\mathbf{%
\mathbb{C}
}^{N\times N}\mathbf{,}$ defined respectively as
\begin{equation}
\boldsymbol{\Delta}%
\triangleq \limfunc{diag}\left( k_{1}\deg \nolimits_{\text{in}%
}(v_{1}),...,k_{N}\deg \nolimits_{\text{in}}(v_{N})\right) \text{\quad
and\quad }\left[ \mathbf{H}(s)\right] _{ij}\triangleq \left\{
\begin{array}{ll}
0, & \text{if }i=j, \\
k_{i}a_{ij}e^{-s\tau _{ij}}, & \text{if }i\neq j,%
\end{array}%
\right.   \label{def_Delta_and_H_matrices}
\end{equation}%
with $\deg \nolimits_{\text{in}}(v_{i})$ given in (\ref{in-out degree}) (see
Appendix \ref{Appendix_overview-Graph Theory}). Observe that $\mathbf{H}%
(s)\in \mathcal{H}^{N\times N}\mathbf{.}$ We can now provide the main result
of this section, stated in the following lemma.

\begin{lemma}
\label{theorem_stability_of_dynamical_system}Given system (\ref%
{system-translated}), assume that the following conditions are satisfied:%
\vspace{-0.2cm}

\begin{enumerate}
\item[\textbf{b1.}] The initial value functions $\mathbf{\phi }\in \mathcal{C%
}_{\beta }^{1},$ and the solutions $\mathbf{x}[\mathbf{\phi }](t)$ are
bounded; \vspace{-0.2cm}

\item[\textbf{b2.}] The characteristic equation associated to (\ref%
{system-translated})
\begin{equation}
p(s)\triangleq \det \left( s\mathbf{I}+%
\boldsymbol{\Delta}%
-\mathbf{H}(s)\right) =0,  \label{def_characteristic-function}
\end{equation}%
with $%
\boldsymbol{\Delta}%
$ and $\mathbf{H}(s)$ defined in (\ref{def_Delta_and_H_matrices}), has all
roots $\{s_{r}\}_{r}\in $ $\mathbf{%
\mathbb{C}
}_{-},$ with \ at most one \emph{simple} root at $s=0.$\footnote{%
We assume, w.l.o.g., that the roots $\{s_{r}\}$ are arranged in
nonincreasing order with respect to the real part, i.e., $0=\limfunc{Re}%
\{s_{0}\}>\limfunc{Re}\{s_{1}\}\geq \limfunc{Re}\{s_{2}\}\geq ...$.}\vspace{%
-0.4cm}
\end{enumerate}

Then, system (\ref{system-translated}) is marginally stable, i.e., $\forall
\mathbf{\phi }\in \mathcal{C}_{\beta }^{1}$ and \ $\limfunc{Re}%
\{s_{1}\}<c<0, $\ there exist $t_{1}$ and $\alpha ,$ with $%
t_{0}<t_{1}<+\infty $ and $0<\alpha <+\infty ,$ independent of $\mathbf{\phi
,}$ and a vector $\mathbf{x}^{\infty },$ with $\left\Vert \mathbf{x}^{\infty
}\right\Vert <+\infty ,$ such that
\begin{equation}
\left\Vert \mathbf{x}[\mathbf{\phi }](t)-\mathbf{x}^{\infty }\right\Vert
\leq \alpha \left\vert \mathbf{\phi }\right\vert _{s}e^{ct},\qquad \forall
t>t_{1}.  \label{inequality_convergence_rate}
\end{equation}
\end{lemma}

\begin{proof}
Under assumption \textbf{b1}), according to \cite[Theorem 6.5]{Bellman-Cooke}%
, the stability properties of system (\ref{system-translated}) are fully
determined by the roots of the characteristic equation associated to (\ref%
{system-translated}), as detailed next.\footnote{%
Observe that assumption \textbf{b1) }is only sufficient for the existence of
the Laplace transform of the solutions to (\ref{system-translated}).}
Denoting by $\mathcal{M}=\{\limfunc{Re}\{s_{r}\}:p(s_{r})=0\}$ the set of
real parts of the characteristic roots $\{s_{r}\}_{r}$, it follows from \cite%
[Theorem 6.7]{Bellman-Cooke} that, $\forall \mathbf{\phi }\in \mathcal{C}%
_{\beta }^{1}$ and \ $c\notin \mathcal{M},$\ there exist $t_{1}$ and $\alpha
,$ with\ $0<t_{1}<+\infty $ and $0<\alpha <+\infty $ independent on $\mathbf{%
\phi }$, such that
\begin{equation}
\left\Vert \mathbf{x}[\mathbf{\phi }](t)-\lim_{l\rightarrow +\infty
}\dsum\limits_{s_{r\in C_{l}:}\limfunc{Re}\{s_{r}\}>c}\mathbf{p}%
_{r}(t)e^{s_{r}t}\right\Vert \leq \alpha \left\vert \mathbf{\phi }%
\right\vert _{s}e^{ct},\qquad \forall t>t_{1},
\label{formula_on_convergence_bound}
\end{equation}%
where each $\mathbf{p}_{r}$ is a (vectorial) polynomial of degree less than
the multiplicity of $s_{r},$ $C_{l}$ denotes a contour in \ the complex
plane of radius increasing with $l$ and centered around $s=0$ (see \cite[%
Sec. 4.1]{Bellman-Cooke} for more details on how such contours $C_{l}$ need
to be chosen), and the sum in (\ref{formula_on_convergence_bound}) is taken
over all characteristic roots $s_{r}$within the contour $C_{l}$ and to the
right of the line $\limfunc{Re}\{s_{r}\}=c.$ Observe that, since the number
of such roots is finite\footnote{%
The reader who is familiar with Linear Retarded Functional Differential
Equations (LRFDE) may observe that this result comes directly from the fact
that the characteristic equation (\ref{def_characteristic-function}) does
not have neutral roots \cite[Ch. 12]{Bellman-Cooke}.} \cite[Theorem 1.5]%
{Gu-book} (see also \cite[Ch. 6.8 and Ch. 12]{Bellman-Cooke}), the limit in (%
\ref{formula_on_convergence_bound}) is always well-defined.

Using (\ref{formula_on_convergence_bound}), we prove now  that, under \textbf{%
b2),} system (\ref{system-translated}) is marginally stable.
Invoking the property that, for any given $\gamma \in
\mathbb{R}
,$ the number of characteristic roots $s_{r}$ with $\limfunc{Re}%
\{s_{r}\}>\gamma $ is finite, \textbf{\ }one can always choose the constant $%
c\notin \mathcal{M}$ in (\ref{formula_on_convergence_bound}) so that $%
\limfunc{Re}\{s_{1}\}<c<\limfunc{Re}\{s_{0}\}=0,$ which leads to%
\begin{equation}
\left\Vert \mathbf{x}[\mathbf{\phi }](t)-\mathbf{p}_{0}\right\Vert \leq
\alpha \left\vert \mathbf{\phi }\right\vert _{s}e^{ct},\qquad \forall
t>t_{1},  \label{inequality_convergence_rate_proof}
\end{equation}%
where we have explicitly used the assumption that the (possible) root $%
s_{0}=0$ is simple and that there are no roots with positive real part. It
follows from \cite[Corollary 6.2]{Bellman-Cooke} that this condition is also
sufficient to guarantee that $\mathbf{p}_{0}$ is bounded, i.e., $\left\Vert
\mathbf{p}_{0}\right\Vert <+\infty .$ Setting in (\ref%
{inequality_convergence_rate_proof}) $\mathbf{x}^{\infty }=\mathbf{p}_{0},$
we obtain (\ref{inequality_convergence_rate}), which completes the proof.
\end{proof}

\noindent \textbf{Remark 1.} It follows from Lemma \ref%
{theorem_stability_of_dynamical_system} that, under assumptions \textbf{b1)}-%
\textbf{b2)}, all the solutions to (\ref{system-translated}) asymptotically
convergence to the constant vector $\mathbf{x}^{\infty }.$ Moreover, \
equation (\ref{inequality_convergence_rate}) provides an estimate of the
convergence rate of the system: as for systems without delays, the
convergence speed of (\ref{system-translated}) is exponential\footnote{%
We say that $\mathbf{x(}t\mathbf{)}$ converges exponentially toward $\mathbf{%
x}^{\infty }$ with rate $r<0$ if $\left\Vert \mathbf{x}(t)-\mathbf{x}%
^{\infty }\right\Vert \leq O\left( e^{rt}\right) .$} with rate arbitrarily
close to $\limfunc{Re}\{s_{1}\}$ that, in general, depends on network
topology and delays.

\noindent \textbf{Remark 2.} Lemma \ref%
{theorem_stability_of_dynamical_system} generalizes results of \cite%
{Stepan-book, Gu-book}, where the authors provided alternative conditions
for the \emph{asymptotic }stability of a Linear Retarded Differential
Equation (LRDE). Interestingly, Lemma \ref%
{theorem_stability_of_dynamical_system} contains some of the conditions of
\cite{Stepan-book, Gu-book} as a special case: System (\ref%
{system-translated}) is asymptotically stable if all the characteristic
roots of (\ref{def_characteristic-function}) have negative real part [see (%
\ref{inequality_convergence_rate})]. This conclusion is the same as that for
linear unretarded systems (see, e.g., \cite{Kailath-book}).

Moreover, Lemma \ref{theorem_stability_of_dynamical_system} can be easily
generalized to include the cases in which one is interested in
\textquotedblleft oscillatory\textquotedblright\ behaviors of system (\ref%
{system-translated}). It is straightforward to see that \emph{bounded}
oscillations arise if assumption \textbf{b2)} is replaced by the following
condition: All the roots of characteristic equation (\ref%
{def_characteristic-function}) have negative real part and the roots with
zero real part are simple.\vspace{-0.4cm}

\section{Proof of Theorem \protect\ref{Theorem_delay-linear_stability}\label%
{proof_Theorem_delay-linear_stability}}

In the following, for the sake of notation simplicity, we drop the
dependence of the state function from the observation, as this dependence
does not play any role in our proof.

\subsection{Sufficiency}

We prove that, under \textbf{a1)}-\textbf{a3)}, the quasi strong
connectivity of digraph ${\mathscr{G}}$ associated to the network in (\ref%
{linear delayed system}) is a sufficient condition for the system (\ref%
{linear delayed system}) to synchronize and that the synchronized state is
given by (\ref{bias_Theo}). To this end, we organize the proof according to
the following two steps.

We first show that, under \textbf{a1)}-\textbf{a2)} and the quasi-strong
connectivity of ${\mathscr{G},}$ the set of LRFDEs (\ref{linear delayed
system}) admits a solution in the form\vspace{-0.3cm}
\begin{equation}
x_{i}^{\star }(t)={\alpha }t+x_{i,0}^{\star },\qquad i=1,\ldots ,N,
\label{subclass_of_synch_state}
\end{equation}%
if ${\alpha =\omega }^{\star }$, where ${\omega }^{\star }$ is defined in (%
\ref{bias_Theo}) and $\{x_{i,0}^{\star }\}$ are constants that depend in
general on the system parameters and the initial conditions. This guarantees
the existence of the desired synchronized state (cf. Definition \ref%
{Definition_sync-state}). Then, invoking results of Appendix \ref%
{Appendix_overview}, we prove that, under \textbf{a1)}-\textbf{a3) }and the
quasi-strong connectivity of ${\mathscr{G},}$ such a synchronized states is
also globally asymptotically stable (according to Definition \ref%
{Definition_sync-state}).\vspace{-0.2cm}

\subsubsection{Existence of a synchronized state\label{Appendix_existence}}

Let us assume that conditions \textbf{a1)}-\textbf{a2) }are satisfied and
that ${\mathscr{G}=}\{{%
\mathscr{V}%
,%
\mathscr{E}%
}\}$ associated to (\ref{linear delayed system}) is QSC.
The synchronized state in the form (\ref{subclass_of_synch_state}) is a
solution to (\ref{linear delayed system}) if and only if it satisfies
equations (\ref{linear delayed system}) (cf. Appendix \ref{Appendix_overview}%
), i.e., if and only if there exist $\alpha$ and $\{x_{i,0}^{\star }\}$\
such that the following system of linear equations is feasible:\vspace{-0.2cm%
}
\begin{equation}
\frac{c_{i}\Delta \omega _{i}(\alpha )}{K}+\sum_{j\in \mathcal{N}%
_{i}}a_{ij}\left( x_{j,0}^{\star }-x_{i,0}^{\star }\right) =0,\qquad \forall
i=1,\ldots ,N,  \label{system-linear-eq}
\end{equation}%
where%
\begin{equation}
\Delta \omega _{i}(\alpha )\triangleq g_{i}(y_{i})-\alpha \left( 1+\frac{K}{%
c_{i}}\sum_{j\in \mathcal{N}_{i}}a_{ij}\,\tau _{ij}\right) .
\label{Delta_omega}
\end{equation}%
Introducing the weighted Laplacian $\mathbf{L}=\mathbf{L}({\mathscr{G}})$
associated to digraph ${\mathscr{G}}$ (cf. Section \ref%
{Appendix_overview-Graph Theory}), the system in (\ref{system-linear-eq})
can be equivalently rewritten as
\begin{equation}
\mathbf{Lx}_{0}^{\star }=\frac{1}{K}\mathbf{D}_{\mathbf{c}}\Delta \,%
\boldsymbol{\omega }(\alpha ),  \label{sys-linear-eqs}
\end{equation}%
where $\mathbf{x}_{0}^{\star }\triangleq \lbrack x_{1,0}^{\star },\ldots
,x_{N,0}^{\star }]^{T},$ $\mathbf{D}_{\mathbf{c}}\triangleq \limfunc{diag}%
(c_{1},\ldots ,c_{N}),$ and $\Delta \,\boldsymbol{\omega }(\alpha
)\triangleq \lbrack \Delta \,\omega _{1}(\alpha ),\ldots ,\Delta \,\omega
_{N}(\alpha )]^{T},$ with $\Delta \,\omega _{i}(\alpha )$ defined in (\ref%
{Delta_omega}). Observe that, under \textbf{a1)}-\textbf{a2) }and the
quasi-strong connectivity of ${\mathscr{G}}$, the graph Laplacian $\mathbf{L}
$ has the following properties (cf. Corollary \ref{Lemma_spanning-tree}):%
\begin{equation}
\mathrm{rank}(\mathbf{L})=N-1,\quad \mathcal{N}(\mathbf{L})=\limfunc{span}%
\left\{ \mathbf{1}_{N}\right\} ,\quad \text{and}\quad \mathcal{N}(\mathbf{L}%
^{T})=\limfunc{span}\left\{ \mathbf{\gamma }\right\} ,  \label{La-properties}
\end{equation}%
where $\mathcal{N}(\mathbf{L})$ denotes the (right) null-space of $\mathbf{L,%
}$ and $\mathbf{\gamma }$ is a left eigenvector of $\mathbf{L}$
corresponding to the (simple) zero eigenvalue of $\mathbf{L}$, i.e., $%
\mathbf{\gamma }^{T}\mathbf{L}=\mathbf{0}^{T}$.

Assume now that $\alpha$ is fixed. It follows from (\ref{La-properties})
that, for any \emph{given} $\alpha ,$ system (\ref{sys-linear-eqs}) admits a
solution if and only if $\ \mathbf{D}_{\mathbf{c}}\Delta \,\mathbf{\omega }%
(\alpha )\in \limfunc{span}\left\{ \mathbf{L}\right\} .$ Because of (\ref%
{La-properties}), we have
\begin{equation}
\mathbf{D}_{\mathbf{c}}\Delta \,\mathbf{\omega }(\alpha )\in \limfunc{span}%
\left\{ \mathbf{L}\right\} \Leftrightarrow \mathbf{\gamma }^{T}\mathbf{D}_{%
\mathbf{c}}\Delta \,\mathbf{\omega }(\alpha )=0.  \label{Cond_on_alpha}
\end{equation}%
It is easy to check that the value of $\alpha $ that satisfies (\ref%
{Cond_on_alpha}) is $\alpha =\omega ^{\star },$ with $\omega ^{\star }$
defined in (\ref{bias_Theo}). Hence, if $\alpha =\omega ^{\star },$ %
the synchronized state in the form (\ref{subclass_of_synch_state}) is a
solution to (\ref{linear delayed system}), for any given set of $\{\tau
_{ij}\},$ $\{g_{i}\},$ $\{c_{i}\}$, $\{a_{ij}\}$ and $K\neq 0$. The
structure of the left eigenvector $\mathbf{\gamma }$ associated to the zero
eigenvalue of $\mathbf{L}$ as given in (\ref{bias_Theo}) comes directly from
Lemma \ref{Lemma_eigenvector_Laplacian}.\vspace{-0.2cm}

Setting $\alpha =\omega ^{\star }$,
system (\ref{sys-linear-eqs}) admits $\infty ^{1}$ solutions, given by
\begin{equation}
\mathbf{x}_{0}^{\star }=\frac{1}{K}\mathbf{\mathbf{L}^{\sharp }D}_{\mathbf{c}%
}\Delta \,\mathbf{\omega }(\omega ^{\star })+\limfunc{span}\{\mathbf{1}%
_{N}\}\triangleq \overline{\mathbf{x}}_{0}+\limfunc{span}\{\mathbf{1}_{N}\},
\label{solutions_linear_system}
\end{equation}%
where
\begin{equation}
\overline{\mathbf{x}}_{0}\triangleq \frac{1}{K}\mathbf{\mathbf{L}^{\sharp }D}%
_{\mathbf{c}}\Delta \,\mathbf{\omega }(\omega ^{\star }),
\label{minimum-norm-solution}
\end{equation}%
$\Delta \omega _{i}(\omega ^{\star })$ is obtained by (\ref{Delta_omega})
setting $\alpha =\omega ^{\star }$ and $\mathbf{\mathbf{L}}_{\mathbf{A}%
}^{\sharp }$ is the generalized inverse of the Laplacian $\mathbf{\mathbf{L}}
$ \cite{Campbell-Meyer}.

\subsubsection{Global Asymptotic Stability of the Synchronized State\ }

To prove the global asymptotic stability of the synchronized state of system
(\ref{linear delayed system}), whose existence has been proved in Appendix %
\ref{Appendix_existence}, we use the following intermediate result (see
Appendix \ref{Appendix_overview} for the definitions used in the lemma).

\begin{lemma}[{\protect\cite[Theorem 2.2]{boyd-SH}}]
\label{Lemma-Boyd}Let $\mathbf{H}(s)\in \mathcal{H}^{N\times N}$ and $\rho
\left( \mathbf{H}(s)\right) $ denote the spectral radius of $\mathbf{H(}s%
\mathbf{).}$ Then, $\rho \left( \mathbf{H}(s)\right) $ is a subharmonic%
\footnote{%
See, e.g., \cite[Ch. 12]{Rudin-book}, \cite{boyd-SH}, for the definition of
subharmonic function.} bounded (above) function on $\overline{\mathbf{%
\mathbb{C}
}}_{+}.$
\end{lemma}

We first rewrite system (\ref{linear delayed system}) in a more convenient
form, as detailed next. Consider the following change of variables
\begin{equation}
\Psi _{i}(t)\triangleq x_{i}(t)-\left( \omega ^{\star }t+\overline{x}%
_{i,0}\right) ,\qquad i=1,\ldots ,N,  \label{Phi-Transformaiton}
\end{equation}%
where $\omega ^{\star }$ and $\{\overline{x}_{i,0}\}$ are defined in (\ref%
{bias_Theo}) and (\ref{minimum-norm-solution}), respectively, so that the
original system (\ref{linear delayed system}) can be equivalently rewritten
in terms of the new variables $\{\Psi _{i}(t)\}_{i}$ as%
\begin{equation}
\begin{array}{l}
\dot{\Psi}_{i}(t)=\Delta \omega _{i}(\omega ^{\star })+\dfrac{K}{c_{i}}%
\dsum\limits_{j\in \mathcal{N}_{i}}a_{ij}\left( \Psi _{j}(t-\tau _{ij})-\Psi
_{i}(t)+\overline{x}_{j,0}-\overline{x}_{i,0}\right) ,\medskip \\
\Psi _{i}(\vartheta )=\phi _{i}(\vartheta )\triangleq \widetilde{\phi }%
_{i}(\vartheta )-\omega ^{\star }\vartheta -\overline{x}_{i,0},\quad
\vartheta \in \lbrack -\tau ,\ 0],%
\end{array}%
\qquad i=1,\ldots ,N,\quad t\geq 0,  \label{Psi-system}
\end{equation}%
where $\widetilde{\mathbf{\phi }}$ are the initial value functions of the
original system (\ref{linear delayed system}). Using (\ref%
{minimum-norm-solution}) [see also (\ref{system-linear-eq}), with $\alpha
=\omega ^{\star }$], system (\ref{Psi-system}) becomes
\begin{equation}
\begin{array}{l}
\dot{\Psi}_{i}(t)=k_{i}\dsum\limits_{j\in \mathcal{N}_{i}}a_{ij}\left( \Psi
_{j}(t-\tau _{ij})-\Psi _{i}(t)\right) ,\quad \quad \quad \vspace{-0.2cm}%
\medskip \smallskip \\
\Psi _{i}(\vartheta )=\phi _{i}(\vartheta ),\quad \vartheta \in \lbrack
-\tau ,\ 0],%
\end{array}%
\text{ \ \ }i=1,\ldots ,N,\quad t\geq 0,  \label{system-translated_2}
\end{equation}%
where, for the sake of convenience, we defined $k_{i}\triangleq K/c_{i}>0,$
for $i=1,\ldots ,N.$

It follows from (\ref{system-translated_2}) that the synchronized state of
system (\ref{linear delayed system}), as given in (\ref%
{subclass_of_synch_state}), is globally asymptotically stable (according to
Definition \ref{Definition_sync-state}) if system in (\ref%
{system-translated_2}) is marginally stable (cf. Appendix \ref%
{Appendix_overview}). To prove the marginal stability of system (\ref%
{system-translated_2}), it is sufficient to show that system (\ref%
{system-translated_2}) satisfies Lemma \ref%
{theorem_stability_of_dynamical_system}. To this end, we organize the rest
of the proof in the \ following two steps.

\begin{description}
\item[Step 1.] We show that, under \textbf{a1)}, \textbf{a3),} all the
solutions $\mathbf{\Psi }[\mathbf{\phi }](t)$ to system (\ref%
{system-translated_2}) are uniformly bounded, as required by assumption
\textbf{b1) }of\textbf{\ }Lemma \ref{theorem_stability_of_dynamical_system};

\item[Step 2.] We prove that, under \textbf{a1)}-\textbf{a3)} and the
quasi-strong connectivity of the digraph, the characteristic equation
associated to system (\ref{system-translated_2}) has all the roots in $%
\mathbb{C}
_{-}$ and a simple root in $s=0,$ which satisfies assumption \textbf{b2) }of%
\textbf{\ }Lemma \ref{theorem_stability_of_dynamical_system}.
\end{description}

\noindent \textbf{Step 1. }Given any arbitrary $\beta <+\infty ,$
assumptions \textbf{a1)}, \textbf{a3)}\ are sufficient to guarantee that all
the trajectories of (\ref{system-translated_2}) are uniformly bounded (cf.
Definition \ref{Def:bounded}), as shown next. Since $\mathbf{\phi }\in
\mathcal{C}_{\beta }^{1},$ we have (see (\ref{C_beta}))%
\begin{equation}
\left\vert \Psi _{i}\left( \vartheta \right) \right\vert \leq \beta ,\qquad
\forall i=1,\ldots ,N,\quad \vartheta \in \lbrack -\tau ,\ 0].
\label{Boundness_init_cond}
\end{equation}%
Condition (\ref{Boundness_init_cond}) is sufficient for $\{\Psi _{i}[\mathbf{%
\phi }](t)\}$ to be uniformly bounded for all $t>0.$ In fact, assume that $%
\{\Psi _{i}[\mathbf{\phi }](t)\}$ are not bounded. Then, according to
Definition \ref{Def:bounded} of Appendix \ref{Appendix_overview}, there must
exist some $\overline{t}>0$ and a set $\mathcal{J}\subseteq \{1,\ldots ,N\}$
such that\footnote{%
For the sake of notation, we omit in the following the dependence of $%
\mathbf{\Psi }[\mathbf{\phi }](t)$ on $\mathbf{\phi .}$}
\begin{equation}
\left\Vert \mathbf{\Psi }(t)\right\Vert _{\infty }\leq \beta ,\quad \forall
t<\overline{t},  \label{Phi_less_beta}
\end{equation}%
and
\begin{equation}
|\Psi _{j}\left( \overline{t}\right) |=\beta \quad \text{and}\quad \overset{%
\cdot }{\Psi }_{j}(\overline{t})=\left\{
\begin{array}{ll}
<0, & \text{if }\Psi _{j}\left( \overline{t}\right) =-\beta , \\
>0, & \text{if }\Psi _{j}\left( \overline{t}\right) =\beta ,%
\end{array}%
\right. \quad j\in \mathcal{J}\text{.}  \label{uniformly-boundness-cond_2}
\end{equation}%
This would imply that, for some $t>\overline{t},$ $|\Psi _{j}\left(
\overline{t}\right) |>\beta .$

Given (\ref{Phi_less_beta})-(\ref{uniformly-boundness-cond_2}) and $j\in
\mathcal{J}$, we have two possibilities, namely:

\begin{description}
\item[i)] $\Psi _{j}\left( \overline{t}\right) =\beta$ and $\overset{\cdot }{%
\Psi }_{j}(\overline{t})>0.$ Since $|\Psi _{i}(\overline{t}-\tau
_{ji})|<\beta ,$ $\forall i=1,\ldots ,N$ (see (\ref{Phi_less_beta})) and $%
k_{j}>0$, from (\ref{system-translated_2}) we have
\begin{equation}  \label{Phi_ineq}
\overset{\cdot }{\Psi }_{j}(\overline{t})=k_{j}\sum_{i\in \mathcal{N}%
_{j}}a_{ji}\left( \Psi _{i}(\overline{t}-\tau _{ji})-\Psi _{j}\left(
\overline{t}\right) \right) =k_{j}\sum_{i\in \mathcal{N}_{j}}a_{ji}\left(
\Psi _{i}(\overline{t}-\tau _{ji})-\beta \right) \leq 0,
\end{equation}%
where in the last inequality, we used (\ref{Boundness_init_cond}).
Expression (\ref{Phi_ineq}) contradicts the assumption $\overset{\cdot }{%
\Psi }_{j}(\overline{t})>0.$

\item[ii)] $\Psi _{j}\left( \overline{t}\right) =-\beta$ and $\overset{\cdot
}{\Psi }_{j}(\overline{t})<0.$ However, since
\begin{equation}
\overset{\cdot }{\Psi }_{j}(\overline{t})=k_{j}\sum_{i\in \mathcal{N}%
_{j}}a_{ji}\left( \Psi _{i}(\overline{t}-\tau _{ji})-\Psi _{j}(\overline{t}%
)\right) =k_{j}\sum_{i\in \mathcal{N}_{j}}a_{ji}\left( \Psi _{i}(\overline{t}%
-\tau _{ji})+\beta \right) \geq 0,
\end{equation}%
a contradiction with $\overset{\cdot }{\Psi }_{j}(\overline{t})<0$ results.
\end{description}

Thus, any solution $\mathbf{\Psi }[\mathbf{\phi }](t)$ to (\ref%
{system-translated_2}) is bounded and remains in $\mathcal{C}_{\beta }^{1}$
for any $t>0.$

\noindent \textbf{Step 2. } We study now the characteristic equation\ (\ref%
{def_characteristic-function}) associated to system (\ref%
{system-translated_2}), assuming that \textbf{a1)}-\textbf{a3)} are
satisfied and that digraph ${\mathscr{G}}$ is QSC. First of all, observe
that, since $\boldsymbol{\Delta}-\mathbf{H}(0)=K \mathbf{D}_{c} \mathbf{L}$,
we have
\begin{equation}
p(0)=\det \left(
\boldsymbol{\Delta}%
-\mathbf{H}(0)\right) =K\det \left( \mathbf{D}_{c}\right) \det \left(
\mathbf{L}\right) =0,  \label{zero_solution}
\end{equation}%
where $%
\boldsymbol{\Delta}%
$ and $\mathbf{H}(s)$ are defined in (\ref{def_Delta_and_H_matrices}), and
the last equality in (\ref{zero_solution}) is due to the properties of
Laplacian matrix $\mathbf{L}=\mathbf{L}({\mathscr{G}})$ of digraph ${%
\mathscr{G}}$ (see (\ref{zero-eig})). It follows from (\ref{zero_solution})
that $p(s)$ has a root in $s=0,$ corresponding to the zero eigenvalue of the
Laplacian $\mathbf{L}$ (recall that $K\det \left( \mathbf{D}_{c}\right) \neq
0$). Since the digraph is assumed to be QSC, according to Corollary \ref%
{Lemma_spanning-tree}, such a root is simple.

Thus, to complete the proof, we need to show that $p(s)$ does not have any
solution in $\overline{\mathbf{%
\mathbb{C}
}}_{+}\backslash \{0\},$ i.e.,%
\begin{equation}
\det \left( s\mathbf{I+}%
\boldsymbol{\Delta}%
-\mathbf{H}(s)\right) \neq 0,\quad \forall s\in \overline{\mathbf{%
\mathbb{C}
}}_{+}\backslash \{0\}.  \label{cond_zeros_in_C_r}
\end{equation}%
Since $s\mathbf{I+}%
\boldsymbol{\Delta}%
$ is nonsingular in $\overline{\mathbf{%
\mathbb{C}
}}_{+}\backslash \{0\}$ [recall that, under \textbf{a1)}, $%
\boldsymbol{\Delta}%
\geq \mathbf{0}$, with at least one positive diagonal entry], (\ref%
{cond_zeros_in_C_r}) is equivalent to
\begin{equation}
\det \left( \mathbf{I-}\left( s\mathbf{I+}%
\boldsymbol{\Delta}%
\right) ^{-1}\mathbf{H}(s)\right) \neq 0,\quad \forall s\in \overline{%
\mathbf{%
\mathbb{C}
}}_{+}\backslash \{0\},
\end{equation}%
which leads to the following sufficient condition for (\ref%
{cond_zeros_in_C_r}):%
\begin{equation}
\rho \left( s\right) \triangleq \rho \left( \left( s\mathbf{I+}%
\boldsymbol{\Delta}%
\right) ^{-1}\mathbf{H}(s)\right) <1,\quad \forall s\in \overline{\mathbf{%
\mathbb{C}
}}_{+}\backslash \{0\}.  \label{cond_spectrum_less_one}
\end{equation}%
Since $\left( s\mathbf{I+}%
\boldsymbol{\Delta}%
\right) ^{-1}\in \mathcal{H}^{N\times N}$ and $\mathbf{H}(s)\in \mathcal{H}%
^{N\times N}$ (cf. Appendix \ref{Appendix_overview}), it follows from Lemma %
\ref{Lemma-Boyd} that the spectral radius $\rho \left( s\right) $ in (\ref%
{cond_spectrum_less_one}) is a subharmonic function on $\overline{\mathbf{%
\mathbb{C}
}}_{+}.$ As a direct consequence, we have, among all, that $\rho \left(
s\right) $ is a continuous bounded function on $\overline{\mathbf{%
\mathbb{C}
}}_{+}$\footnote{%
Observe that $\rho \left( s\right) $ is well-defined in $s=0,$ and $\rho
\left( 0\right) =1.$} and satisfies the \emph{maximum modulus principle }%
(see, e.g., \cite[Ch. 12]{Rudin-book}): $\rho \left( s\right) $ achieves its
\emph{global }maximum only on the boundary of $\overline{\mathbf{%
\mathbb{C}
}}_{+}.$\footnote{%
According to the maximum modulus theorem \cite[Theorem 10.24]{Rudin-book},
the only possibility for $\rho \left( s\right) $ to reach its global maximum
also on the interior of $\overline{\mathbf{%
\mathbb{C}
}}_{+}$ is that $\rho \left( s\right) $ be constant over all $\overline{%
\mathbf{%
\mathbb{C}
}}_{+}$, which is not the case.} Since $\rho \left( s\right) $ is strictly
proper in $\overline{\mathbf{%
\mathbb{C}
}}_{+},$ i.e., $\rho \left( s\right) \rightarrow 0$ as $\left\vert
s\right\vert \rightarrow +\infty $ while keeping $s\in \overline{\mathbf{%
\mathbb{C}
}}_{+}$, it follows that
\begin{equation}
\sup_{s\in \mathbf{%
\mathbb{C}
}_{+}}\rho \left( s\right) <\sup_{s\in \overline{\mathbf{%
\mathbb{C}
}}_{+}}\rho \left( s\right) \leq \sup_{\omega \in
\mathbb{R}
}\rho \left( j\omega \right) .  \label{sup_spectral_radius}
\end{equation}%
Using (\ref{sup_spectral_radius}), we infer that condition (\ref%
{cond_spectrum_less_one}) is satisfied if
\begin{equation}
\rho \left( j\omega \right) =\rho \left( \left( j\omega \mathbf{I+}%
\boldsymbol{\Delta}%
\right) ^{-1}\mathbf{H}(j\omega )\right) <1,\quad \forall \omega \in
\mathbb{R}
\backslash \{0\}.
\end{equation}%
For any matrix norm $\left\Vert \mathbf{\cdot }\right\Vert $, since $\rho
\left( \mathbf{A}\right) \leq \left\Vert \mathbf{A}\right\Vert$ $\forall
\mathbf{A\in
\mathbb{C}
}^{N\times M}$ \cite[Theorem 5.6.9]{Horn-book}, using the \emph{maximum row
sum matrix norm } $\left\Vert \mathbf{\cdot }\right\Vert _{\infty }$ defined
as \cite[Definition 5.6.5]{Horn-book}%
\begin{equation}
\left\Vert \mathbf{A}\right\Vert _{\infty }\triangleq \max_{r=1,\ldots
,N}\dsum\limits_{q=1}^{M}\left\vert \left[ \mathbf{A}\right]
_{rq}\right\vert ,
\end{equation}%
we have%
\begin{eqnarray}
\rho \left( j\omega \right) &\leq &\left\Vert \left( j\omega \mathbf{I+}%
\boldsymbol{\Delta}%
\right) ^{-1}\mathbf{H}(j\omega )\right\Vert _{\infty }=\max_{r=1,\ldots
,N}\dsum\limits_{q\neq r}\left\vert \frac{k_{r}\text{ }a_{rq}}{j\omega
+k_{r}\deg \nolimits_{\text{in}}(v_{r})}e^{-j\omega \tau _{rq}}\right\vert
\notag \\
&= &\max_{r=1,\ldots ,N}\left\vert \frac{k_{r}\deg \nolimits_{\text{in}%
}(v_{r})}{j\omega +k_{r}\deg \nolimits_{\text{in}}(v_{r})}\right\vert \leq 1,
\label{spectrum_inequality_on_omega}
\end{eqnarray}%
where in the last inequality the equality is reached if and only if $\omega
=0.$ It follows from (\ref{spectrum_inequality_on_omega}) that $\rho \left(
j\omega \right) <1$ for all $\omega \neq 0,$ which guarantees that condition
(\ref{cond_spectrum_less_one}) is satisfied.

This proves that assumption \textbf{b2) }of Lemma \ref%
{theorem_stability_of_dynamical_system} holds true. Hence, all the
trajectories $\mathbf{\Psi }[\mathbf{\phi }](t)\rightarrow \mathbf{\Psi }%
^{\infty }$ as $t\rightarrow +\infty ,$ with exponential rate arbitrarily
close to $r\triangleq \{\min_i \limfunc{Re}\{s_{i}\}: p(s_{i})=0\,\, \text{%
and}\,\, s_i\neq 0\}$, where $p(s)$ is defined in (\ref%
{def_characteristic-function}) and $\mathbf{\Psi }^{\infty }$ satisfies the
linear system of equations $\mathbf{L\Psi }^{\infty }=\mathbf{0}$, whose
solution is $\mathbf{\Psi }^{\infty }\in \limfunc{span}\{\mathbf{1}_{N}\}$
(Corollary \ref{Lemma_spanning-tree}). In other words, system (\ref%
{system-translated_2}) exponentially reaches the consensus on the state.
\hspace{\fill}\rule{1.5ex}{1.5ex}%
\vspace{-0.2cm}

\subsection{Necessity}

We prove the necessity of the condition by showing that, if the digraph ${%
\mathscr{G}}$ of (\ref{linear delayed system}) is not QSC, different
clusters of nodes synchronize on different values. This local
synchronization is in contrast with the definition of (global)
synchronization, as given in Definition \ref{Definition_sync-state}. Hence,
if the overall system has to synchronize, the digraph associated to the
system must be QSC.

Assume that the digraph ${\mathscr{G}}$ associated to (\ref{linear delayed
system}) is not QSC, but WC with $K$ SCCs and, let us say, $r\leq K$
(distinct) RSCCs . Then, according to Lemma \ref{condensation-digraph}, the
condensation digraph ${\mathscr{G}}^{\star }{=}\{{%
\mathscr{V}%
}^{\star }{,%
\mathscr{E}%
}^{\star }\}$ contains a spanning directed forest with $r$ distinct roots
(associated to the $r$ RSCCs of the $K$ SCCs of ${\mathscr{G}}$)$.$ Ordering
the nodes $v_{1}^{\star },\ldots ,v_{K}^{\star }\in {%
\mathscr{V}%
}^{\star }$ according to Lemma \ref{condensation-digraph-vertex-ordering},
and exploring the relationship between ${\mathscr{G}}^{\star }$and ${%
\mathscr{G}}$ (cf. Appendix \ref{Appendix_existence}), one can write the
Laplacian matrix $\mathbf{L}=\mathbf{L}({\mathscr{G}})$ as an $r$-reducible
matrix \cite{Horn-book}, i.e., 
\begin{equation}
\mathbf{L(}\mathscr{G}\mathbf{)}=\left(
\begin{array}{cccccc}
\mathbf{L}_{1} & \mathbf{0} & \cdots & \cdots & \cdots & \mathbf{0} \\
\mathbf{0} & \ddots & \ddots & \mathbf{\cdots } & \mathbf{\cdots } & \vdots
\\
\mathbf{0} & \mathbf{0} & \mathbf{L}_{r} & \ddots & \mathbf{\cdots } & \vdots
\\
\ast & \ast & \ast & \mathbf{B}_{r+1} & \ddots & \vdots \\
\ast & \ast & \ast & \ast & \ddots & \mathbf{0} \\
\ast & \ast & \ast & \ast & \ast & \mathbf{B}_{K}%
\end{array}%
\right) ,  \label{L_r_forest}
\end{equation}%
where the first $r$ diagonal blocks are the Laplacian matrices of the $r$
RSCCs of ${\mathscr{G},}$ and the $\{\mathbf{B}_{k}\},$ with $k>r,$ are the
nonsingular matrices associated to the remaining SCCs. Each of these
matrices can be written as the linear combination of the Laplacian matrix of
the corresponding SCC and a nonnegative diagonal matrix with at least one
positive diagonal entry. The structure of $\mathbf{L}$ given in (\ref%
{L_r_forest}) shows that the RSCCs associated to the first $r$ diagonal
blocks are totally decoupled from each other. Hence, (at least) the state
derivatives of the nodes in each of these $r$ RSCCs reach a common value
(since the corresponding subdigraphs are SC by construction) that, in
general, is different for any of the SCCs. This is sufficient for the
overall system not to reach a global synchronization.

\def\baselinestretch{1}

\end{document}